\documentclass[11pt]{article}

\usepackage[margin=1in]{geometry}

\usepackage{amsmath}
\usepackage{amssymb}
\usepackage{amsthm}
\usepackage{mathtools}
\usepackage{xcolor}

\usepackage{graphicx}
\usepackage{booktabs}
\usepackage{comment}

\usepackage{tikz}
\usetikzlibrary{decorations.markings}

\usepackage[algo2e,linesnumbered,ruled,vlined]{algorithm2e}

\usepackage{float}

\usepackage{comment}
\usepackage{booktabs}

\usepackage{authblk}

\usepackage[authoryear]{natbib}

\usepackage[algo2e,linesnumbered,ruled,vlined]{algorithm2e}
\usetikzlibrary{decorations.markings}

\usepackage[hidelinks]{hyperref}
\usepackage{cleveref}

\DeclareMathOperator{\Cov}{Cov}
\DeclareMathOperator{\blkdiag}{blkdiag}
\newcommand{\proper}{\mathsf}

\newcommand{\pE}{\proper{E}}
\newcommand{\pV}{\proper{V}}

\newcommand{\pN}{\proper{N}}
\newcommand{\Var}{\pV}
\newcommand{\mv}[1]{{\boldsymbol{\mathrm{#1}}}}

\theoremstyle{plain}
\newtheorem{theorem}{Theorem}
\newtheorem{lemma}{Lemma}
\newtheorem{proposition}{Proposition}
\newtheorem{corollary}{Corollary}

\theoremstyle{definition}
\newtheorem{definition}{Definition}
\newtheorem{assumption}{Assumption}
\newtheorem{example}{Example}

\theoremstyle{remark}
\newtheorem{remark}{Remark}

\newcommand{\Bin}{\beta}   
\newcommand{\Mtr}{\mv M}   
\newcommand{\inoutframe}{%
  \draw[gray] (-1.35,-1.35) -- (-1.35,1.25);
  \draw[gray] (-1.35,-1.35) -- (1.25,-1.35);
  \foreach \t in {-1,0,1}{
    \draw[gray] (-1.35,\t) -- ++(-0.05,0)
      node[left, inner sep=1.5pt, font=\scriptsize] {$\t$};
    \draw[gray] (\t,-1.35) -- ++(0,-0.05)
      node[below, inner sep=1.5pt, font=\scriptsize] {$\t$};}
  \node[font=\small] at (-0.05,-1.72) {$x$};
  \node[font=\small, rotate=90] at (-1.78,-0.05) {$y$};}
\tikzset{
  vtx/.style={draw, circle, inner sep=1pt},
  elab/.style={inner sep=2pt},
  gph/.style={x=1.38cm,y=1.38cm,
    decoration={markings, mark=at position 0.5 with {\arrow{>}}}}}

\title{Gaussian Processes on Directed Metric Graphs}

\author{
David Bolin\textsuperscript{1},
Alexandre B. Simas\textsuperscript{1},
Erik Karlsson Strandh\textsuperscript{2,*},
and Jonas Wallin\textsuperscript{2}
}

\date{}

\begin{document}

\maketitle

\footnotetext[1]{%
Statistics Program, CEMSE Division,
King Abdullah University of Science and Technology,
23955-6900 Thuwal, Saudi Arabia.}

\footnotetext[2]{%
Department of Statistics,
Lund University,
SE-220 07 Lund, Sweden.}

\begingroup
\renewcommand{\thefootnote}{*}
\footnotetext{%
Corresponding author:
Erik Karlsson Strandh,
Department of Statistics, Lund University,
Box 743, SE-220 07 Lund, Sweden.
E-mail:
\href{mailto:erik.karlsson_strandh@stat.lu.se}
{erik.karlsson\_strandh@stat.lu.se}.}
\endgroup

\begin{abstract}
We introduce a statistical framework for Gaussian fields indexed at arbitrary edge locations on general compact directed metric graphs. The construction is based on a stochastic differential equation with a first-order operator and conditions at the vertices. We characterise well-posedness and identify the covariance reproducing kernel Hilbert space. 
We also connect the proposed framework to earlier stream-network models, showing that these arise from the same system under particular boundary conditions, and introduce new boundary conditions that yield more physically realistic processes.
The differential-equation representation enables computationally efficient inference and prediction. This makes the method applicable to large data sets without approximation. Applications to
temperature modelling on river networks and traffic speeds on road networks
illustrate the framework, including the computational efficiency and improved performance under physically informed vertex conditions.
\end{abstract}

\vspace{0.5em}
\noindent\textbf{Keywords:}
directed metric graphs, Gaussian processes, graph reduction, Ornstein--Uhlenbeck processes, reproducing kernel Hilbert spaces, river networks.


\section{Introduction}\label{sec:intro}

Data observed on physical networks arise in
transport, infrastructure, and hydrology. In many such networks, physical flow determines the direction of transport: traffic density evolves along directed road links, while constituents are carried through pipe networks and mixed at junctions \citep{HoldenRisebro1995,ShangEtAl2021}. 

Continuously indexed directional Gaussian fields have been introduced on flow-oriented trees by moving averages \citep{Hoef06,Hoef10,PeterE10}, and spatio-temporal extensions have been considered \citep{santosfernandez2022ssnbayesrpackagebayesian}.
However, these continuously indexed models have not been introduced on general compact metric graphs, 
and they are computationally expensive. To address the latter, approximations such as domain partitioning and mesh-based constructions have been considered \citep{rivArt,THORSON2019143}. 
Directionality beyond trees has also been studied at the discrete graph level.
 \citet{MaddixEtAl2022} construct a vertex-indexed Mat\'ern Gaussian process on general finite directed weighted graphs using a graph-advection operator.

In earlier work, we introduced symmetric Whittle--Mat\'ern fields on a compact metric graph $\Gamma$
\citep{bolin2024gaussian} as solutions to the fractional-order equation
\begin{equation}\label{eq:SPDE}
    (\kappa^2 - \Delta_{\Gamma})^{\alpha/2}(\tau u) = \dot{W} \qquad\text{on $\Gamma$.}
\end{equation}
Here $\Delta_\Gamma$ is the Kirchhoff Laplacian, which acts as the second derivative along edges while enforcing continuity and a sum-to-zero condition on outward derivatives at vertices; $\dot W$ is white noise; $\kappa$ controls the correlation range, $\alpha>1/2$ the sample-path regularity, and $\tau>0$ is a variance-scale parameter. For integer $\alpha$,
Markov properties enable exact sparse inference and prediction
\citep{bolin2023statistical,bolin2026markov} for these models, which are implemented in
\texttt{MetricGraph} \citep{MetricGraphpackage}. 
The Kirchhoff vertex conditions used to define $\Delta_{\Gamma}$ are, however, symmetric, whereas transport along
a directed network need not be. 


The goal of this work is to introduce a directed version of these Whittle--Mat\'ern fields, which then defines a class of directed models on general metric graphs, and not only on trees, without specifying a
directed covariance directly. We specifically study the first-order equation
\begin{equation}
\label{eq:SDE}
    (\kappa + \partial_{\Gamma})(\tau u) = \dot{W},
\end{equation}
where $\partial_{\Gamma}$ acts as the derivative along the edges and the equation is equipped with linear constraints at the vertices.
The covariance operator of $u$ in \eqref{eq:SPDE} is $\tau^{-2}(\kappa^2 - \Delta_{\Gamma})^{-\alpha}$, whereas $u$ in \eqref{eq:SDE} heuristically has covariance operator
$\tau^{-2}(\kappa + \partial_{\Gamma})^{-1}(\kappa - \partial_{\Gamma})^{-1} = \tau^{-2}(\kappa^2 - \widetilde{\Delta}_{\Gamma})^{-1}$,
where $\widetilde{\Delta}_{\Gamma}$ is a Laplacian with vertex conditions induced by those of $\partial_{\Gamma}$. Thus, the solution to \eqref{eq:SDE} is essentially a Whittle--Mat\'ern field with $\alpha = 1$ but with different vertex conditions. %
The advantage with \eqref{eq:SDE} is that it facilitates imposing vertex conditions which are more suitable for directed networks with flow-driven dependence.  On trees, particular choices recover the tail-up constructions of \citet{Hoef06}, tail-down constructions of \citet{Hoef10}, and further that one can also define other more physically realistic vertex conditions.

This unified model class defines, to our knowledge, the first statistical framework for Gaussian fields at arbitrary edge locations on general directed metric graphs. 
Building on the classical correspondence \citep{kimeldorf1970correspondence} between processes and reproducing kernel Hilbert spaces (RKHS), we identify the covariance RKHS and show that we can obtain exact solutions, without the need for discretizations. The construction is a linear structural equation model (SEM) whose variables are indexed by the points of $\Gamma$ rather than by a finite vertex set. We make the connection to SEMs exact by deriving a continuous version of the trek rule for SEMs \citep[Thm.~4.1]{Drton2018} for these models.
Finally, we show that Markov properties of the model class facilitate exact and highly computationally efficient inference that allows us to analyze much larger data sets than what has previously been possible using exact models.

The remainder of the paper is organised as follows.
\Cref{sec:construction} introduces the model class and \Cref{sec:rkhs} establishes well-posedness, derives the RKHS, and gives an explicit representation of the field.
\Cref{sec:acyclic} derives a closed-form covariance kernel for acyclic graphs and \Cref{seq:Constraint_comparison} relates the
construction to symmetric Whittle--Mat\'ern fields and to the tail-up and tail-down models. \Cref{sec:inference} develops the inference procedures, and \Cref{sec:applications} presents the two applications.
Proofs and implementation details are collected in appendices in the supplementary materials.

\section{Directed metric graphs and directed Gaussian processes}\label{sec:construction}

\subsection{Directed metric graphs}
\label{sec:graphs}

Let $\Gamma=(\mathcal V,\mathcal E)$ be a finite connected metric graph
obtained by gluing endpoints of intervals $[0,\ell_e]$, $0<\ell_e<\infty$.
The processes considered here are continuous along edges, but their traces
need not agree at vertices. We therefore use the split-edge index set
$\widetilde\Gamma:=\bigsqcup_{e\in\mathcal E}
(\{e\}\times[0,\ell_e])$, with quotient map
$\pi:\widetilde\Gamma\to\Gamma$. Each nonvertex point has a unique lift
$(e,t)$, whereas $\pi^{-1}(v)$ is the set of labelled incident edge ends at $v$; see \Cref{fig:inout}(c). Set
$\ell_{\min}:=\min_e\ell_e>0$ and fix vertex and edge orderings for all
indexed vectors and matrices. Write $d(x,y)$ for the path distance between
their projections on the underlying undirected metric graph, omitting the
projections when $x,y\in\Gamma$.

Define the vertex-valued maps
$\operatorname{tail},\operatorname{head}:\mathcal E\to\mathcal V$ by $\operatorname{tail}(e):=\pi(e,0)$ and $ \operatorname{head}(e):=\pi(e,\ell_e)$. For $v\in\mathcal V$, let $\mathcal E_v^{\mathrm{in}} :=\{e:\operatorname{head}(e)=v\}$ and $ \mathcal E_v^{\mathrm{out}} :=\{e:\operatorname{tail}(e)=v\}$, and introduce the corresponding labelled endpoint sets $\mathcal V_v^{\mathrm{in}} :=\{(e,\ell_e):e\in\mathcal E_v^{\mathrm{in}}\}$ and $\mathcal V_v^{\mathrm{out}}:=\{(e,0):e\in\mathcal E_v^{\mathrm{out}}\}$. Thus $\widetilde{\mathcal V}_v:=\pi^{-1}(v) = \mathcal V_v^{\mathrm{in}}\cup \mathcal V_v^{\mathrm{out}}$ and $\deg(v)=|\widetilde{\mathcal V}_v|$.
This endpoint representation is the directed analogue of the quantum-graph construction of a vertex through the collection of edge ends incident to it \cite[Section~1.4.1]{BerkolaikoKuchment2013}.

The graph is \emph{acyclic} if it has no directed cycle and a \emph{directed
tree} if its underlying undirected graph is a tree. A directed route follows
increasing edge coordinates through projected vertices; write $x\leadsto y$
if one exists, allowing the zero-length route, so $x\leadsto x$. The sources
are ${\mathcal V_-:=\{v:\mathcal E_v^{\mathrm{in}}=\varnothing\}}$. Put
$m:=|\mathcal V_-|$, possibly zero, and, when $m>0$, write
$\mathcal V_-=\{s_{-,1},\dots,s_{-,m}\}$. The source edges are
$\mathcal E_-:=\{e:\operatorname{tail}(e)\in\mathcal V_-\}$. A source is an
\emph{inward leaf} if it has one outgoing edge.

\begin{assumption}\label{ass:standing}
Throughout, $\Gamma$ is a finite connected directed metric graph in which
every source is an inward leaf.
\end{assumption}

This is a modelling restriction that avoids specifying a joint Gaussian law
for several outgoing source traces; shared, independent and correlated
initial values give different extensions. We use $\Gamma$ for topology and
paths and $\widetilde\Gamma$ for evaluation, with $f(e,t):=f_e(t)$. A field
descends to $\Gamma$ precisely when $f(\xi)=f(\xi')$ for every
$\xi,\xi'\in\pi^{-1}(v)$ and every $v\in\mathcal V$.
For $v\in\mathcal V_-$, let $e_v$ be its sole outgoing edge and set
$f(v):=f_{e_v}(0)$ and $u(v):=u_{e_v}(0)$; elsewhere vertex notation
denotes only a common trace.
On an outgoing incidence $f_e(v)$ means $f_e(0)$ and on an incoming incidence
it means $f_e(\ell_e)$, so the two ends of a self-loop remain distinct.

For $e\in\mathcal E$, let $L_2(e)$ and $H^1(e)$ be the usual spaces on
$[0,\ell_e]$, and set
\[
L_2(\Gamma)=\bigoplus_{e\in\mathcal E}L_2(e),
\qquad
\widetilde H^1(\Gamma)=\bigoplus_{e\in\mathcal E}H^1(e),
\]
with corresponding direct-sum norms.  They are spaces on $\widetilde\Gamma$, and elements
of $\widetilde H^1(\Gamma)$ may have distinct traces above a vertex.

Throughout, we use $a(\cdot,\cdot)$ for symmetric bilinear
forms, which may be positive semidefinite, and reserve
$\langle\cdot,\cdot\rangle$ for genuine inner products.

\subsection{The model on a single edge}
\label{sec:edge}

On a single edge $e=[0,\ell_e]$, let $L_e:=\kappa I+d/dt$, $\kappa>0$.
Then \eqref{eq:SDE} becomes $L_e(\tau u)=\dot W_e$, equivalently
\begin{equation}
\label{eq:edge-sde}
du(t)=-\kappa u(t)\,dt+\tau^{-1}dW_e(t),
\qquad t\in(0,\ell_e),
\end{equation}
where $W_e$ is a standard Brownian motion. Since
$\ker(L_e)=\operatorname{span}\{e^{-\kappa t}\}$, an initial value is
required. For $u_0\sim\pN(0,\sigma_0^2)$,
$0<\sigma_0^2<\infty$, independent of $W_e$, the unique solution is the Ornstein--Uhlenbeck process 
\begin{equation}
    u(t)
    =
    u_0e^{-\kappa t}
    +
    \tau^{-1}\int_0^t e^{-\kappa(t-s)}\,dW_e(s),
    \qquad 0\leq t\leq\ell_e.
    \label{eq:OU}
\end{equation}
Thus $\kappa^{-1}$ is the exponential decay length. Setting
$
\sigma_0^2=\sigma^2:=\frac{1}{2\kappa\tau^{2}}
$
makes $u$ stationary with marginal variance $\sigma^2$; we call
this \emph{stationary anchoring}. Under this anchoring,
$\kappa^{-1}$ is also the correlation length. 


The Cameron--Martin space of \eqref{eq:OU} is the building block for the
graph model. For $f,g\in H^1(e)$, define the edgewise
positive semidefinite bilinear form
\begin{equation}
    a_e^L(f,g)
    :=
    \int_0^{\ell_e}(L_e f)(t)(L_e g)(t)\,dt.
    \label{eq:QuadForm}
\end{equation}

\begin{proposition}\label{lem:innerProdAndRKHS}
Let $\kappa,\tau>0$ and $0<\sigma_0^2<\infty$.  The process \eqref{eq:OU}
has Cameron--Martin space $H^1(e)$ with inner product
\[
    \langle f,g\rangle_{\mathcal H_u}
    =
    \tau^{2}a_e^L(f,g)+\sigma_0^{-2}f(0)g(0).
\]
\end{proposition}
For $g=f$, the two terms are the noise-input and initial-value
energies, respectively.
For comparison, let $\Delta_{N,e}$ be the Neumann Laplacian and write
\[
a_e^A(f,g):=\int_0^{\ell_e}\{f'g'+\kappa^2fg\}\,dt.
\]
The $\alpha=1$ Whittle--Mat\'ern field solving
$(\kappa^2I-\Delta_{N,e})^{1/2}(\tau u^A)=\dot W_e$ has Cameron--Martin
space $H^1(e)$ and inner product
$\langle f,g\rangle_{\mathcal H_{u^A}}=\tau^2a_e^A(f,g)$
\citep{bolin2024gaussian}.

\begin{lemma}\label{lem:edge-ips}
For $f,g\in H^1(e)$,
    $a_e^L(f,g)
    =
    a_e^A(f,g)+
    \kappa\{f(\ell_e)g(\ell_e)-f(0)g(0)\}$.
\end{lemma}

Thus the directed and symmetric forms differ only at the endpoints. Under
stationary anchoring, \Cref{lem:innerProdAndRKHS} and
\Cref{lem:edge-ips} give
\begin{equation*}
    \langle f,g\rangle_{\mathcal H_u}
    =
    \tau^2\bigl\{
        a_e^A(f,g)+\kappa f(0)g(0)+\kappa f(\ell_e)g(\ell_e)
    \bigr\},
\end{equation*}
so the stationary directed edge model is the symmetric model with a soft
anchor of weight $\kappa\tau^2$ at each endpoint. 

\subsection{The model on a metric graph}
\label{sec:model}

Extend $L_e$ to $\Gamma$ edgewise by
$(\kappa+\partial_\Gamma)u:=\bigoplus_{e\in\mathcal E}L_eu_e$, so that
$\kappa+\partial_\Gamma$ acts locally along edges and all coupling
between edges is carried by conditions at the vertices. Given coefficients
$\Bin$, seek a centred Gaussian field $u=(u_e)_{e\in\mathcal E}$ satisfying
\begin{equation}
\label{eq:system}
\left\{
\begin{aligned}
L_e(\tau u_e)&=\dot W_e,
&& e\in\mathcal E,\\
u_e(v)&=\sum_{\hat e\in\mathcal E^{\mathrm{in}}_v}
\Bin_v(e,\hat e)\,u_{\hat e}(v),
&& v\in\mathcal V\setminus\mathcal V_-,\
e\in\mathcal E^{\mathrm{out}}_v,\\
u(v)&\sim\pN(0,\sigma_v^2),
&& v\in\mathcal V_-,
\end{aligned}
\right.
\end{equation}
where $\{W_e\}_{e\in\mathcal E}$ are independent standard Brownian
motions, the source values $\{u(v)\}_{v\in\mathcal V_-}$ are mutually
independent and independent of $\{W_e\}_{e\in\mathcal E}$, and
$0<\sigma_v^2<\infty$. Note that on a cyclic graph, \eqref{eq:system} is a stochastic
boundary-value problem, and not a causal It\^o evolution. 

Because $\kappa+\partial_\Gamma$ is first order, coupling is specified through endpoint values without separately imposing the derivative-matching or Kirchhoff flux condition of the standard symmetric continuous Whittle--Mat\'ern construction. This is why the directed formulation admits vertex conditions that the second-order Kirchhoff formulation cannot, and it is
the reason for the added flexibility claimed in \Cref{sec:intro}. 

We call a family $\Bin=(\Bin_v)_{v\notin\mathcal V_-}$ of coefficients as
in \eqref{eq:system} a \emph{forward} vertex condition, since it
generates each outgoing value from the incoming ones, and write
\[
\widetilde H^1_{\Bin}(\Gamma)
:=
\Bigl\{f\in\widetilde H^1(\Gamma):
f_e(v)=\!\!\sum_{\hat e\in\mathcal E^{\mathrm{in}}_v}\!\!\Bin_v(e,\hat e)f_{\hat e}(v)
\ \text{ for all } v\notin\mathcal V_-,\ e\in\mathcal E^{\mathrm{out}}_v\Bigr\}
\]
for the corresponding subspace of $\widetilde H^1(\Gamma)$. Not
every vertex condition of interest is forward. For a non-forward condition,
the second line of \eqref{eq:system} is replaced by its full linear trace
constraint, and the field is then instead defined through the energy form in \Cref{sec:rkhs}. We now discuss the specific
vertex conditions further.

\subsection{Vertex conditions}
\label{sec:vertexconditions}

For a nonsource $v$, fix $w_{v,\hat e}>0$, $\hat e\in\mathcal E_v^{\mathrm{in}}$,
and write
\[
p_{v,\hat e}:=\frac{w_{v,\hat e}}{\sum_{\tilde e\in\mathcal E^{\mathrm{in}}_v}w_{v,\tilde e}},
\qquad
\sum_{\hat e\in\mathcal E^{\mathrm{in}}_v}p_{v,\hat e}=1.
\]
In hydrological applications, $w_{v,\hat e}$ may be discharge or a proxy such as stream width or drainage area, and if the weights are proportional to discharge then $p_{v,\hat e}$ is a flow proportion. 
We consider three types of vertex conditions. The $K_1$ and $K_2$ rows below apply at every nonsource vertex, whereas the $C_V$ row applies only when $|\mathcal E_v^{\mathrm{in}}|=1$:
\begin{equation}
\label{eq:vertex-condition}
\Bin_v(e,\hat e)=
\begin{cases}
1, & \text{continuity, } C_V,\\[-5.1pt]
p_{v,\hat e}, & \text{flow-weighted coupling, } K_1,\\[-5.1pt]
\sqrt{p_{v,\hat e}}, & \text{square-root-weighted coupling, } K_2.
\end{cases}
\end{equation}
Continuity ($C_V$) is standard for symmetric metric-graph fields \citep{bolin2024gaussian}, and square-root weighting ($K_2$) is used in tail-up models \citep{Hoef06}. The remaining condition ($K_1$) is, to our knowledge, new for Gaussian processes on networks. 
Here $K_1$ and $K_2$ are forward conditions, and we
write $\widetilde H^1_{K_1}(\Gamma)$ and $\widetilde H^1_{K_2}(\Gamma)$ for the spaces $\widetilde H^1_{\Bin}(\Gamma)$ they determine. 
At a confluence, $C_V$ also
equates the incoming traces and is imposed directly through the space
\begin{equation*}
\widetilde H^1_{C_V}(\Gamma)
:=\bigl\{f\in\widetilde H^1(\Gamma):
f(\xi)=f(\xi')\ \text{for all }\xi,\xi'\in
\widetilde{\mathcal V}_v,\ v\in\mathcal V\bigr\}.
\end{equation*}
At a sink, $K_1$ and $K_2$ impose no constraint, whereas
$C_V$ still equates multiple incoming traces. See \Cref{app:matrix-bc} for the matrix forms of these conditions.

\subsection{What the vertex conditions do}
\label{sec:whatcondsdo}

\begin{figure}[H]
  \centering
\begin{tabular}{c@{\hskip 1em}c@{\hskip 1em}c}
\begin{tikzpicture}[gph]
  \inoutframe
  \node[vtx] (v1) at (-1,-1) {$v_1$};
  \node[vtx] (v2) at (0,0)   {$v_2$};
  \node[vtx] (v3) at (1,-1)  {$v_3$};
  \node[vtx] (v4) at (0,1)   {$v_4$};
  \draw[postaction={decorate}] (v1) -- (v2) node[midway, above left,  elab] {$e_2$};
  \draw[postaction={decorate}] (v3) -- (v2) node[midway, above right, elab] {$e_3$};
  \draw[postaction={decorate}] (v2) -- (v4) node[midway, right,       elab] {$e_1$};
\end{tikzpicture}
&
\begin{tikzpicture}[gph]
  \inoutframe
  \node[vtx] (v1) at (-1,-1) {$v_1$};
  \node[vtx] (v2) at (0,0)   {$v_2$};
  \node[vtx] (v3) at (1,-1)  {$v_3$};
  \node[vtx] (v4) at (0,1)   {$v_4$};
  \draw[postaction={decorate}] (v2) -- (v1) node[midway, above left,  elab] {$e_2$};
  \draw[postaction={decorate}] (v2) -- (v3) node[midway, above right, elab] {$e_3$};
  \draw[postaction={decorate}] (v4) -- (v2) node[midway, right,       elab] {$e_1$};
\end{tikzpicture}
&
\begin{tikzpicture}[gph]
  \inoutframe
  \node[vtx] (v1) at (-1,-1)      {$v_1$};
  \node[vtx] (v3) at (1,-1)       {$v_3$};
  \node[vtx] (v4) at (0,1)        {$v_4$};
  \node[vtx, font=\small] (a2) at (-0.24,-0.24) {$v_{e_2}$};
  \node[vtx, font=\small] (a3) at (0.24,-0.24)  {$v_{e_3}$};
  \node[vtx, font=\small] (a1) at (0,0.22)      {$v_{e_1}$};
  \draw[postaction={decorate}] (v1) -- (a2) node[midway, above left,  elab] {$e_2$};
  \draw[postaction={decorate}] (v3) -- (a3) node[midway, above right, elab] {$e_3$};
  \draw[postaction={decorate}] (a1) -- (v4) node[midway, right,       elab] {$e_1$};
\end{tikzpicture}
\\
{\small (a) $\Gamma_{\mathrm{in}}$} & {\small (b) $\Gamma_{\mathrm{out}}$} & {\small (c) split of $\Gamma_{\mathrm{in}}$}
\end{tabular}
  \caption{The confluence $\Gamma_{\mathrm{in}}$, the
  divergence $\Gamma_{\mathrm{out}}$, and the labelled endpoint copies of
  $\Gamma_{\mathrm{in}}$.}
  \label{fig:inout}
\end{figure}
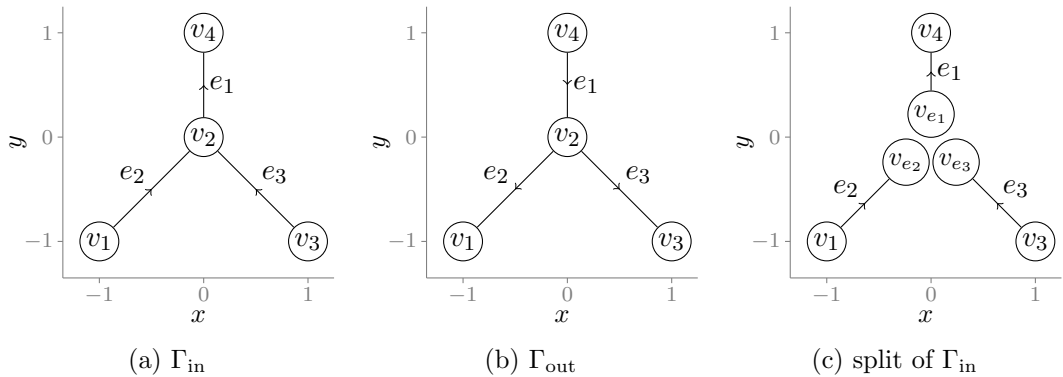
\Cref{fig:inout}(a) is the simplest confluence at which $K_1$
and $K_2$ produce different outgoing traces. In panel (b), all three rules
coincide because there is one inflow, and panel (c) shows the labelled endpoint
copies. On $\Gamma_{\mathrm{in}}$ with unit weights,
\begin{align*}
K_1 &: u_{e_1}(v_2)
=\tfrac12u_{e_2}(v_2)+\tfrac12u_{e_3}(v_2),\\
K_2 &: u_{e_1}(v_2)
=\sqrt{\tfrac12}u_{e_2}(v_2)+\sqrt{\tfrac12}u_{e_3}(v_2).
\end{align*}
If the incoming traces share a value $z$, $K_1$ returns $z$ whereas $K_2$
returns $\sqrt2z$. If they are independent with common variance $\sigma^2$,
the corresponding outgoing variances are $\sigma^2/2$ and $\sigma^2$. Thus
$K_1$ preserves a common incoming value, while $K_2$ preserves marginal
variance under the stated assumptions. Neither condition imposes continuity
at a general confluence. 

The variance-preserving property of $K_2$ is not special to this example:
\Cref{cor:stationary-variance-K2} shows that on a directed tree, $K_2$
with stationary anchoring propagates a constant marginal variance through
the whole network, which is the behaviour built into the tail-up models
of \citet{Hoef06}; see \Cref{sec:tailupdown}.

When $w_{v,\hat e}$ is discharge, $K_1$ has a direct
conservation interpretation. For pollutant concentration or temperature,
the transported mass or heat flux is proportional to
$w_{v,\hat e}u_{\hat e}(v)$. At a one-outflow confluence, conservation gives
$u_e(v)=\sum_{\hat e}p_{v,\hat e}u_{\hat e}(v)$,
which is exactly $K_1$. With several outflows, this is the common perfectly
mixed value when total outflow equals total inflow. The square-root rule has
no analogous conservation interpretation, and conditional means are in general discontinuous under $K_2$ whereas they are continuous under $K_1$, see 
\Cref{fig:simple_directional}.

\begin{figure}[H]
\centering
\includegraphics*[width=0.49\linewidth]{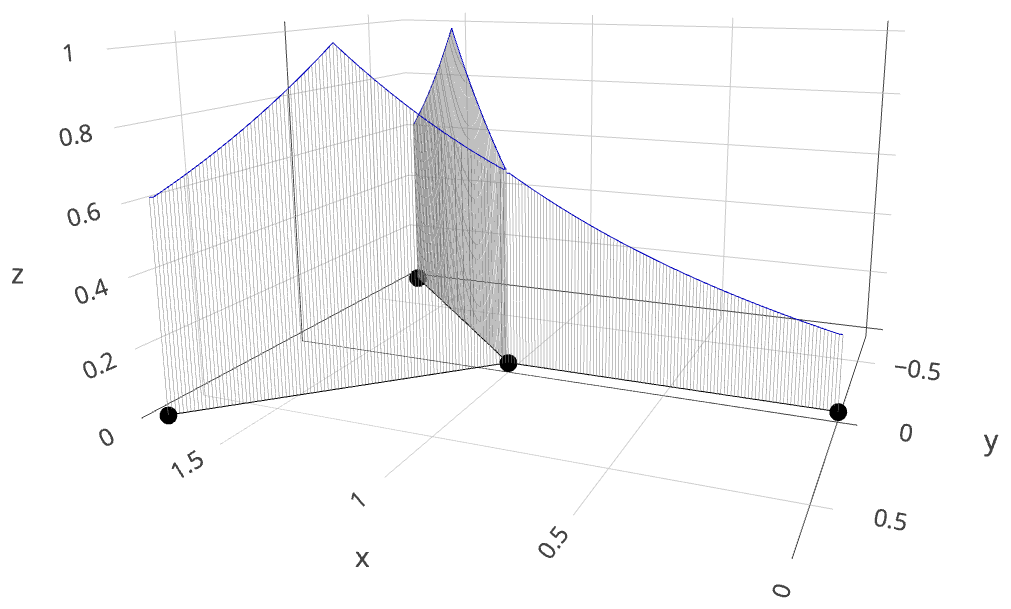}\hfill
\includegraphics*[width=0.49\linewidth]{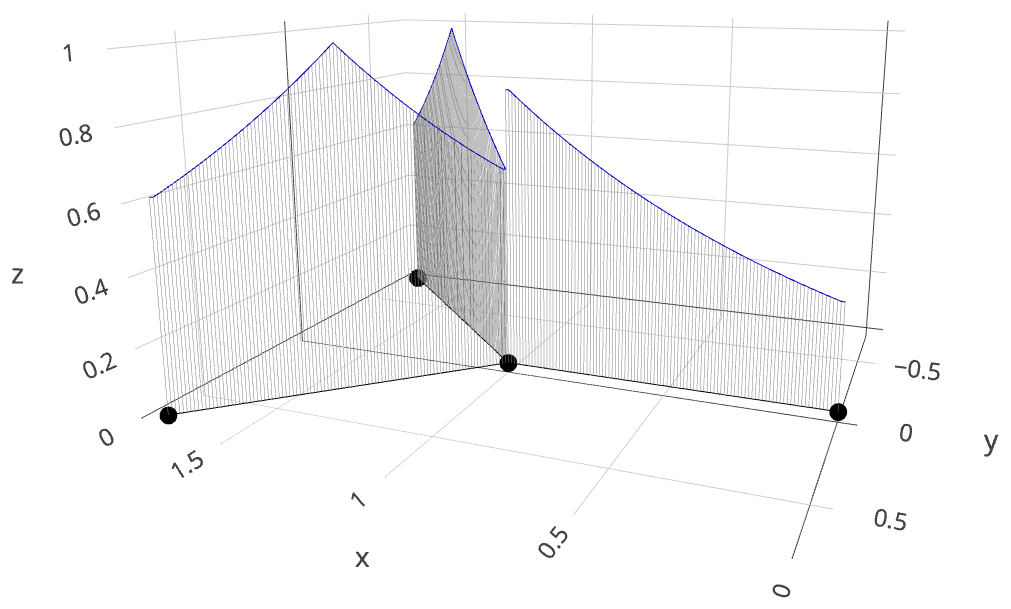}
\caption{Conditional means on $\Gamma_{\mathrm{in}}$ given
$u_{e_2}(0.5)=u_{e_3}(0.5)=1$, under flow-weighted $K_1$ (left) and
square-root-weighted $K_2$ (right). The former has matching traces at $v_2$,
whereas the latter has an amplified outgoing trace.}
\label{fig:simple_directional}
\end{figure}

\section{Well-posedness and the reproducing kernel Hilbert space}
\label{sec:rkhs}

The edge equations and source anchors in \eqref{eq:system} induce an energy form on
$\widetilde H^1(\Gamma)$. In this section we show that its restriction to a vertex-condition space is
positive definite exactly when the corresponding model is well posed, and we call
such graph--condition pairs admissible. In that case the form is the covariance
RKHS inner product, and we obtain an explicit representation of the process.

\subsection{The energy form and well-posedness}
\label{sec:energyform}
\label{sec:wellposed}

For $f,g\in\widetilde H^1(\Gamma)$ and strictly positive
anchoring coefficients $(c_v)_{v\in\mathcal V_-}$, define
\begin{equation}
    a_\Gamma^0(f,g)
    :=
    \sum_{v\in\mathcal V_-} c_v\, f(v)g(v)
    +
    \tau^{2}
    \sum_{e\in\mathcal E}
    a_e^L(f_e,g_e).
    \label{eq:q_semi}
\end{equation}
The two sums in \eqref{eq:q_semi} are the source-anchoring
and edgewise noise-input energies as in \Cref{lem:innerProdAndRKHS}. The correspondence with 
\eqref{eq:system} is $c_v=\sigma_v^{-2}$, while stationary anchoring gives
$c_v\equiv2\kappa\tau^2$. 
On $\widetilde H^1(\Gamma)$ this form is only positive semidefinite, with null
space given by zero source traces and $L_ef_e=0$ on every edge. The model is
well posed precisely when its vertex conditions eliminate this null space.

For the random field $u$, set $\eta_e:=u_e(0)$. Solving the
edge equation with integrating factor $e^{\kappa t}$ gives
\begin{equation}
\label{eq:edgewise-rep}
u_e(t)=e^{-\kappa t}\eta_e+\tau^{-1}\!\!\int_0^t
e^{-\kappa(t-s)}\,dW_e(s),
\qquad t\in[0,\ell_e].
\end{equation}
This is \eqref{eq:OU} with initial value $\eta_e$. Define
$\zeta_e:=\int_0^{\ell_e}e^{-\kappa(\ell_e-s)}\,dW_e(s)$ so that 
$\zeta_e\sim\pN\bigl(0,(1-e^{-2\kappa\ell_e})/(2\kappa)\bigr)$ are independent across edges,
and $u_e(\ell_e)=e^{-\kappa\ell_e}\eta_e+\tau^{-1}\zeta_e$. 
Substituting the terminal traces into a forward vertex condition $\Bin$ gives $\eta_e=(\Mtr\eta)_e+\xi_e$, where $\Mtr$ is the \emph{vertex transfer matrix}
\begin{equation}
\label{eq:transfer-matrix}
\Mtr_{e,\hat e}
:=
\begin{cases}
\Bin_v(e,\hat e)\,e^{-\kappa\ell_{\hat e}},
& v:=\operatorname{tail}(e)\notin\mathcal V_-
\text{ and }\hat e\in\mathcal E^{\mathrm{in}}_v,\\
0,&\text{otherwise,}
\end{cases}.
\end{equation}
Here $\xi_e$ collects the source and noise contributions. We write $\rho(\Mtr)$ for the spectral radius of $\Mtr$ and note that $\rho(\Mtr)<1$ guarantees that $I-\Mtr$ is invertible, but the converse need not hold.

\begin{lemma}
\label{lem:q_is_inner_product}
Let \Cref{ass:standing} hold, let $\kappa,\tau>0$, and let
the anchoring coefficients in \eqref{eq:q_semi} be strictly positive. For
$X\in\{C_V,K_1,K_2\}$, the form $a_\Gamma^0$ is positive definite on
$\widetilde H^1_X(\Gamma)$ if $\Gamma$ is acyclic or
$X\in\{C_V,K_1\}$. For a forward condition $\Bin$, it is positive definite
on $\widetilde H^1_{\Bin}(\Gamma)$ if and only if $I-\Mtr$ is invertible.
\end{lemma}

For $K_1$, the proof uses only the absolute row-sum bound
$\sum_{\hat e}|\Bin_v(e,\hat e)|\leq1$ for every outgoing edge $e$, and
therefore applies to any forward condition with this property. The bound
fails for $K_2$ at a vertex with at least two inflows, because
$\sum_{\hat e}\sqrt{p_{v,\hat e}}>
\sum_{\hat e}p_{v,\hat e}=1$. This is harmless on an acyclic graph but can
prevent admissibility on a cyclic graph, as \Cref{ex:cyclic_K2_null_mode}
shows.

\begin{example}
\label{ex:cyclic_K2_null_mode}
On the graph in \Cref{fig:cyclic_K2_null_mode}, assign incoming weights
$(1,3)$ to $(e_1,e_4)$ at $v$ and equal weights to $(e_2,e_3)$ at $w$. A
null vector of $I-\Mtr$ has $\eta_1=0$, because $l$ is a source, and
satisfies
\[
\eta_2=\eta_3=\frac{\sqrt3}{2}\eta_4e^{-\kappa\ell_4},
\qquad
\eta_4=\frac{\eta_2e^{-\kappa\ell_2}+\eta_3e^{-\kappa\ell_3}}{\sqrt2}.
\]
Set $\ell_2=\ell_3$,
$\ell_2+\ell_4=(2\kappa)^{-1}\log(3/2)$, and $\eta_4=1$. Then
$\eta_2=\eta_3=(\sqrt3/2)e^{-\kappa\ell_4}\neq0$, so $1$ is an eigenvalue
of $\Mtr$ and, by \Cref{lem:q_is_inner_product},
$a_\Gamma^0$ is not positive definite.
\end{example}
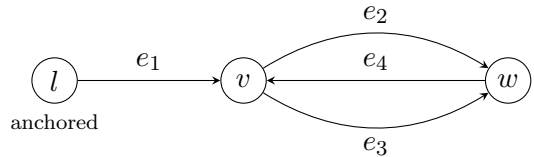
\begin{figure}[H]
\centering
\begin{tikzpicture}[
    vertex/.style={circle,draw,minimum size=6mm,inner sep=1pt},
    edge/.style={->,>=stealth}
]
    \node[vertex] (s) at (0,0) {$l$};
    \node[vertex] (v) at (2.5,0) {$v$};
    \node[vertex] (w) at (6,0) {$w$};
    \draw[edge] (s) -- node[above] {$e_1$} (v);
    \draw[edge] (v) to[bend left=32] node[above] {$e_2$} (w);
    \draw[edge] (v) to[bend right=32] node[below] {$e_3$} (w);
    \draw[edge] (w) -- node[above] {$e_4$} (v);
    \node[font=\scriptsize] at (0,-0.55) {anchored};
\end{tikzpicture}
\caption{An anchored source feeding a two-branch directed
cycle. \Cref{ex:cyclic_K2_null_mode} shows that $K_2$ need not be admissible
on a cyclic graph.}
\label{fig:cyclic_K2_null_mode}
\end{figure}

Acyclicity is sufficient but not necessary for admissibility and $K_2$ can be admissible on cyclic graphs, such as the traffic network in
\Cref{sec:applications}. We now show that the energy space is an RKHS, and hence defines a Gaussian field, for every admissible pair $(\Gamma,X)$. This includes the non-forward condition $C_V$ at confluences.

\begin{proposition}
\label{lem:q_equiv_H1_graph}
\label{lem:q_RKHS_graph}
Under \Cref{ass:standing}, let $\kappa,\tau>0$, take
strictly positive anchoring coefficients in \eqref{eq:q_semi}, and let
$X\in\{C_V,K_1,K_2\}$. If $a_\Gamma^0$ is positive definite on
$\widetilde H^1_X(\Gamma)$, then
\(\langle f,g\rangle_{\Gamma,X}:=a_\Gamma^0(f,g)\)
is an inner product whose norm is equivalent to the ambient
$\widetilde H^1(\Gamma)$-norm, and
$\bigl(\widetilde H^1_X(\Gamma),
\langle\cdot,\cdot\rangle_{\Gamma,X}\bigr)$ is an RKHS
on $\widetilde\Gamma$.
\end{proposition}

\subsection{Explicit representation}
\label{sec:transfer}
\label{sec:subdivision}

For a forward condition, the following proposition solves
\eqref{eq:system} on an arbitrary directed metric graph and identifies the
Cameron--Martin space of the resulting field with the RKHS in
\Cref{lem:q_RKHS_graph}.

\begin{proposition}\label{prop:transfer-construction}
Let \Cref{ass:standing} hold, let $\Bin$ be a forward
condition, and suppose $I-\Mtr$ is invertible. Then \eqref{eq:system} has a
unique solution, given by \eqref{eq:edgewise-rep} with
$\eta=\Mtr\eta+\xi$, or equivalently
\begin{equation}
\label{eq:fixedpoint}
\eta=(I-\Mtr)^{-1}\xi,
\qquad
\xi_e=
\begin{cases}
u(\operatorname{tail}(e)), & e\in\mathcal E_-,\\[2pt]
\tau^{-1}\!\!\sum_{\hat e\in\mathcal E^{\mathrm{in}}_v}\!\!\Bin_v(e,\hat e)\,\zeta_{\hat e},
& \text{otherwise,}
\end{cases}
\end{equation}
where $v=\operatorname{tail}(e)$. The field $u$ is centred
Gaussian and for distinct $e,e'$, the components $\xi_e$ and $\xi_{e'}$ are
independent unless the two edges leave the same interior vertex. The
Cameron--Martin space is $\widetilde H^1_{\Bin}(\Gamma)$ with inner product
$\langle f,g\rangle_{\mathcal H_u}:=a_\Gamma^0(f,g)$ and anchoring
coefficients $c_v=\sigma_v^{-2}$, $v\in\mathcal V_-$.
\end{proposition}

Locality of $L_e$ also makes the law invariant under
subdivision.

\begin{proposition}
\label{prop:subdivision}
Let $X\in\{C_V,K_1,K_2\}$, and obtain $\Gamma_2$ from
$\Gamma_1$ by inserting a degree-two vertex at an interior point
$(e,t_0)$. Then
$a^0_{\Gamma_1}=a^0_{\Gamma_2}$ under the natural identification
of $\widetilde H^1_X(\Gamma_1)$ and $\widetilde H^1_X(\Gamma_2)$.
\end{proposition}

Whenever the model is admissible, subdividing an edge leaves the law of
the centred Gaussian field unchanged under this identification. The same
holds when removing a vertex with one incoming and one outgoing edge.

\section{The acyclic case: recursion, proper OU processes and covariances}
\label{sec:acyclic}

\subsection{Forward recursion and proper OU processes}
\label{sec:forwardrecursion}
On an acyclic graph, the field can be represented as a system
of Ornstein--Uhlenbeck processes generated recursively from its sources. We
call this graph-level construction a proper global OU process.
\begin{definition}
\label{def:proper-global-OU}
Let $\Gamma$ be a finite acyclic directed metric graph
satisfying \Cref{ass:standing}, and let $\kappa,\tau>0$. A centred Gaussian
field $u=(u_e)_{e\in\mathcal E}$ indexed by
$\widetilde\Gamma$ is a \emph{proper global OU process} if
there are mutually independent standard Brownian motions
$\{W_e\}_{e\in\mathcal E}$ and centred Gaussian initial values
$\{\eta_e\}_{e\in\mathcal E}$ such that
\[
u_e(t)=e^{-\kappa t}\eta_e+\tau^{-1}\!\!\int_0^t e^{-\kappa(t-s)}\,dW_e(s),
\qquad t\in[0,\ell_e],\ e\in\mathcal E,
\]
where $\{\eta_e\}_{e\in\mathcal E_-}$ are mutually
independent, have finite positive variances, and are independent of
$\{W_e\}_{e\in\mathcal E}$, while, for every $e\notin\mathcal E_-$,
$\eta_e$ is measurable with respect to
$\sigma\bigl(u_{\hat e}(\ell_{\hat e}):
\hat e\in\mathcal E^{\mathrm{in}}_{\operatorname{tail}(e)}\bigr)$.
\end{definition}

The representation \eqref{eq:edgewise-rep} solves the
first-order equation on each edge. On an acyclic graph, $\eta_e$ depends only
on source values and upstream noises, and is therefore independent of $W_e$.
The next two results show that the forward construction in
\Cref{sec:transfer} characterizes all proper global OU processes.

\begin{lemma}
\label{lem:graph-OU-form}
Let $u$ be a proper global OU process on a finite acyclic
directed metric graph $\Gamma$ satisfying \Cref{ass:standing}. For every
$v\notin\mathcal V_-$, there is a deterministic matrix
$B_v\in\mathbb R^{|\mathcal E^{\mathrm{out}}_v|\times|\mathcal E^{\mathrm{in}}_v|}$
such that
$\bigl(\eta_e\bigr)_{e\in\mathcal E^{\mathrm{out}}_v}
=B_v\bigl(u_{\hat e}(\ell_{\hat e})\bigr)_{
\hat e\in\mathcal E^{\mathrm{in}}_v}$ almost surely. Hence $u$ solves
\eqref{eq:system} with the forward condition
${\Bin_v(e,\hat e)=(B_v)_{e,\hat e}}$. For each source
$v\in\mathcal V_-$, set $\sigma_v^2:=\Var(\eta_{e_v})$. The Cameron--Martin space of $u$ is
$\widetilde H^1_{\Bin}(\Gamma)$ with inner product
\[
\langle f,g\rangle_{\Gamma}^{\mathrm{OU}}
=
\tau^2\sum_{e\in\mathcal E}a_e^L(f_e,g_e)
+
\sum_{v\in\mathcal V_-}\sigma_v^{-2}\,f(v)g(v),
\qquad f,g\in\widetilde H^1_{\Bin}(\Gamma).
\]
\end{lemma}

Consequently, with $c_v=\sigma_v^{-2}$ for
$v\in\mathcal V_-$, the Cameron--Martin inner product in
\Cref{lem:graph-OU-form} satisfies
$\langle f,g\rangle_{\Gamma}^{\mathrm{OU}}=a_\Gamma^0(f,g)$, with
$a_\Gamma^0$ defined in \eqref{eq:q_semi}.

\begin{corollary}
\label{cor:forward-substitution}
Let $\Gamma$ be acyclic and $\Bin$ a forward vertex
condition. Then $\Mtr^{|\mathcal E|}=0$, the field in
\Cref{prop:transfer-construction} is a proper global OU process with
$B_v=(\Bin_v(e,\hat e))$, and
\[
    (I-\Mtr)^{-1}
    =\sum_{k=0}^{|\mathcal E|-1}\Mtr^k.
\]
\end{corollary}

Thus \eqref{eq:fixedpoint} is solved by forward substitution
in any topological edge ordering. Together with \Cref{lem:graph-OU-form},
this shows that the forward construction and the class of proper global OU
processes coincide on acyclic graphs. In particular, $K_1$ and $K_2$ yield
proper global OU processes. The same holds for $C_V$ when every interior
vertex has in-degree one. At a confluence, however, continuity also constrains the incoming traces; it is therefore not forward and does not give a proper global OU process.

\subsection{Covariance functions}
\label{sec:covariances}

The covariance structures of \citet{Hoef06} are obtained from
moving-average constructions on stream networks.  Here the covariance is
induced by \eqref{eq:system} instead, and is built from a \emph{transfer
factor} describing how a value propagates downstream.  Throughout this
subsection $u$ is a proper global OU process on $\Gamma$.
We additionally assume that $B_v$ has identical rows at
every interior vertex, as all three conditions in
\eqref{eq:vertex-condition} do whenever they are forward. Thus
$\Bin_v(e,\hat e)$ is independent of the outgoing edge $e$, and the outgoing
traces have a common value $u(v)$. Incoming traces remain edge-labelled
unless continuity is imposed.  Here and only here, $u(v)$ denotes that common
outgoing trace.

Using the labelled endpoint sets of \Cref{sec:graphs}, write
$\bar v=(\hat e,\ell_{\hat e})\in\mathcal V_v^{\mathrm{in}}$ and
$\underline v=(e,0)\in\mathcal V_v^{\mathrm{out}}$, and abbreviate
$\Bin_v(\underline v,\bar v):=\Bin_v(e,\hat e)$.  Then the vertex relation is
$u(\underline v)=\sum_{\bar v\in\mathcal V_v^{\mathrm{in}}}
\Bin_v(\underline v,\bar v)u(\bar v)$ almost surely.
The covariance formulas below use locations in $\Gamma^\circ:=\Gamma\setminus\mathcal V$, which have unique lifts to $\widetilde\Gamma$; labelled endpoint covariances follow by the corresponding one-sided edge limits.
For $x,y\in\Gamma^\circ$, let $\mathcal P(x,y)$ be the directed routes from $x$ to $y$; use the unique source lift or common outgoing state when an endpoint is a source or interior vertex. A route $\rho$ has length $\mathrm{len}(\rho)$ and routed-vertex set $P^\circ(\rho)$, which includes a terminal common-outgoing state but excludes an initial one. For $v\in P^\circ(\rho)$, $\bar v(\rho)$ and $\underline v(\rho)$ are its incoming and outgoing labelled ends.
Define the transfer
factor
\begin{equation}
\label{eq:transferFactor}
\mathsf A(x,y):=\sum_{\rho\in\mathcal P(x,y)}
\exp\bigl(-\kappa\,\mathrm{len}(\rho)\bigr)
\prod_{v\in P^\circ(\rho)}\Bin_v\bigl(\underline v(\rho),\bar v(\rho)\bigr),
\end{equation}
with the empty product equal to $1$ and the empty sum to $0$, so that
$\mathsf A(x,y)=0$ unless $x\leadsto y$.
Thus $\mathsf A$ combines edgewise decay with the routing
coefficients along every path. On an acyclic graph the sum is finite and,
after subdivision at interior endpoints, agrees with the corresponding
Neumann expansion of $(I-\Mtr)^{-1}$ from \Cref{sec:transfer}.
On a tree, let $\rho_{x,y}$ be the unique directed path, when it exists, and set $P^\circ(x,y):=P^\circ(\rho_{x,y})$, $\bar v(x,y):=\bar v(\rho_{x,y})$, and $\underline v(x,y):=\underline v(\rho_{x,y})$. Then \eqref{eq:transferFactor} reduces to
\[
\mathsf A(x,y)=\exp\bigl(-\kappa d(x,y)\bigr)
\prod_{v\in P^\circ(x,y)}\Bin_v\bigl(\underline v(x,y),\bar v(x,y)\bigr),
\qquad x\leadsto y.
\]

Write $Z_j:=u(s_{-,j})$ for the value of the field at the $j$th source vertex.
For $y\in\Gamma^\circ$ set ${\mathcal J(y):=\{j:s_{-,j}\leadsto y\}}$ and define
the upstream ancestral subgraph
$\Lambda^\uparrow(y):=\{z\in\Gamma:z\leadsto y\}$.
Integrals over such subgraphs are understood edgewise; the
convention and the resulting variance identities are collected in
\Cref{supp:S3}.

\begin{theorem}\label{thm:cov-kernel}
Let $\Gamma$ be a finite acyclic directed metric graph and $u$ a proper
global OU process on $\Gamma$ such that, for every interior vertex $v$,
the matrix $B_v$ has identical rows, with source variances $\sigma^2_{s_{-,j}}$.
Write $\mathcal C^\uparrow(x,y):=\Lambda^\uparrow(x)\cap\Lambda^\uparrow(y)$
for the common ancestral set.  Then, for all
$x,y\in\Gamma^\circ$, $r(x,y)=\Cov\bigl(u(x),u(y)\bigr)$ satisfies
\begin{equation}
\label{eq:cov-kernel}
r(x,y)=\!\!\sum_{j\in\mathcal J(x)\cap\mathcal J(y)}\!\!
\sigma^2_{s_{-,j}}\,\mathsf A(s_{-,j},x)\,\mathsf A(s_{-,j},y)
+\tau^{-2}\!\!\int\limits_{\mathcal C^\uparrow(x,y)}\!\!
\mathsf A(z,x)\,\mathsf A(z,y)\,dz.
\end{equation}
\end{theorem}

Every driver upstream of both points contributes, weighted by the transfer
factor to each.  Expanding the two factors as sums over paths writes
\eqref{eq:cov-kernel} as a sum over \emph{pairs} of directed paths issuing
from a common ancestor, which is the continuous analogue of
the trek rule for linear structural equation models
\citep[Thm.~4.1]{Drton2018}.

\begin{remark}
\label{rem:cyclic-kernel}
On a cyclic graph the transfer operator remains
$(I-\Mtr)^{-1}$ whenever the inverse exists, and $\rho(|\Mtr|)<1$ guarantees
absolute convergence of its directed-path expansion. We do not claim
\eqref{eq:cov-kernel} in this case, because the stochastic representation and
interchange of the infinite path sums also require justification.
\end{remark}

Taking $x=y$ in \eqref{eq:cov-kernel} gives the general
variance formula, and along one edge it reduces to the usual OU recursion;
both identities are stated in \Cref{cor:variance-formula}.
On a tree equation \eqref{eq:cov-kernel} collapses.  If
$\mathcal C^\uparrow(s,t)\neq\varnothing$ then this set has a unique
maximal element for the partial order $x\preceq y\iff x\leadsto y$, as the
proof of \Cref{thm:cov-lca} shows; we denote it by $a(s,t)$ and call it the
\emph{last common ancestor} of $s$ and $t$.

\begin{corollary}\label{thm:cov-lca}
Let $\Gamma$ be a finite directed metric tree and $u$ a proper global OU
process on $\Gamma$.  Then, for $s,t\in\Gamma^\circ$,
\[
r(s,t)=
\begin{cases}
r(a,a)\,\mathsf A(a,s)\,\mathsf A(a,t),
& \text{if } a:=a(s,t)\text{ exists},\\[0.3em]
0, & \text{if } \mathcal C^\uparrow(s,t)=\varnothing.
\end{cases}
\]
In particular, if $t\leadsto s$ then $a(s,t)=t$ and
$r(s,t)=r(t,t)\,\mathsf A(t,s)$.
\end{corollary}

Thus, on a tree, covariance is the variance at the last common
ancestor propagated to both points, with strength determined by the routing
coefficients. Points without a common ancestor are independent; with strictly
positive coefficients the converse also holds. The next result shows that
$K_2$ preserves the stationary variance
$\sigma^2=(2\kappa\tau^2)^{-1}$.

\begin{corollary}\label{cor:stationary-variance-K2}
Let $\Gamma$ be a finite directed metric tree and let $u$ be a proper
global OU process on $\Gamma$ with vertex condition $K_2$.  If
$\Var\bigl(u(s_{-,j})\bigr)=\sigma^2$ for every $j=1,\dots,m$, then
$r(t,t)=\sigma^2$ for every $t\in\widetilde\Gamma$.
\end{corollary}

\section{Relation to symmetric and stream-network models}
\label{sec:relation}
\label{seq:Constraint_comparison}

We now relate the construction to the symmetric
Whittle--Mat\'ern fields of \citet{bolin2024gaussian}, the tail-up models
\citep{Hoef06,Hoef10}, and the tail-down models \citep{Hoef10}.

\subsection{Symmetric versus directional forms under continuity}
\label{sec:symmvsdir}

By \Cref{lem:edge-ips}, the directed and symmetric edge forms
differ only at their endpoints. On the continuity domain, write
$\langle\cdot,\cdot\rangle^{\mathrm{sym}}_\Gamma:=\tau^2\sum_e a^A_e$ for
the Cameron--Martin inner product of the $\alpha=1$ Whittle--Mat\'ern field.
Summing the endpoint terms gives the following graph-level identity.

\begin{proposition}\label{prop:sym-vs-dir}
Let \Cref{ass:standing} hold and let $f,g\in\widetilde H^1_{C_V}(\Gamma)$.
Then
\begin{equation}
\label{eq:sym-vs-dir}
a_\Gamma^0(f,g)
=\langle f,g\rangle^{\mathrm{sym}}_\Gamma
+\kappa\tau^2\sum_{v\in\mathcal V}
\bigl(|\mathcal E^{\mathrm{in}}_v|-|\mathcal E^{\mathrm{out}}_v|\bigr)f(v)g(v)
+\sum_{v\in\mathcal V_-}c_v\,f(v)g(v).
\end{equation}
\end{proposition}

Thus the directed and symmetric forms differ only through
evaluations at vertices, and the weight at $v$ depends on $\Gamma$ only
through the in- and out-degrees of $v$. Apart from source anchoring, a vertex contributes
$\kappa\tau^2(|\mathcal E_v^{\mathrm{in}}|-|\mathcal
E_v^{\mathrm{out}}|)$. Hence an outward leaf has weight
$\kappa\tau^2$, an inward leaf $-\kappa\tau^2$, and an interior
contribution vanishes exactly when its in- and out-degrees are equal.
Stationary anchoring, $c_v=2\kappa\tau^2$ for
$v\in\mathcal V_-$, changes the total weight at an
inward leaf to $+\kappa\tau^2$.

The star graph $\Gamma_{\mathrm{out}}$ in
\Cref{fig:inout}(b) gives a useful pseudo-observation interpretation. Its
inward leaf is $v_4$, and its junction $v_2$ has one inflow and two outflows, so $C_V$ is forward because the junction has a single inflow. Define the
stationary-anchored directional form and the leaf-penalized symmetric form by
\begin{align*}
\langle f,g\rangle^{\mathrm{dir}}_{\Gamma_{\mathrm{out}}}
&:=\left.a^{0}_{\Gamma_{\mathrm{out}}}(f,g)
\right|_{c=2\kappa\tau^2},\quad
\langle f,g\rangle^{\mathrm{sym},\partial}_{\Gamma_{\mathrm{out}}}
:=\langle f,g\rangle^{\mathrm{sym}}_{\Gamma_{\mathrm{out}}}
+\kappa\tau^2\!\sum_{i\in\{1,3,4\}} f(v_i)g(v_i).
\end{align*}
Then \Cref{prop:sym-vs-dir} reduces to
$
\langle f,g\rangle^{\mathrm{sym},\partial}_{\Gamma_{\mathrm{out}}}
=\langle f,g\rangle^{\mathrm{dir}}_{\Gamma_{\mathrm{out}}}
+\kappa\tau^2 f(v_2)g(v_2).
$
The additional term is exactly the precision contributed by
an independent pseudo-observation at the junction. If
$\mathcal L_{\mathrm{dir}}$ and $\mathcal L_{\mathrm{sym},\partial}$ denote
the two Gaussian laws, then
\[
\mathcal L_{\mathrm{sym},\partial}(u)
=\mathcal L_{\mathrm{dir}}(u\mid Y=0),
\qquad
Y\mid u\sim\pN\bigl(u(v_2),(\kappa\tau^2)^{-1}\bigr).
\]
Thus, apart from their common leaf penalties, the symmetric
field is the directed field conditioned on extra information at the junction, so its marginal variance is smaller at every site correlated with
$u(v_2)$.

\subsection{Tail-up, tail-down and the role of orientation}
\label{sec:tailupdown}

The relation to the tail-up and tail-down models of \citet{Hoef10} is mediated by 
the vertex conditions and the
orientation, which fixes the direction of dependence.  We say that a finite
directed metric tree is \emph{oriented along the flow} if every interior
vertex has exactly one outgoing edge.  Then exactly one vertex has no
outgoing edge, and this vertex is a leaf; we call it the \emph{outlet}.
The sources of $\Gamma$ are the remaining leaves, so a flow-oriented tree
has one outlet and, in general, many sources.

\begin{proposition}[tail-up form]\label{prop:tailup}
Let $\Gamma$ be a finite directed metric tree oriented along the flow, with
vertex condition $K_2$ and stationary anchoring at the sources.  Then
$r(t,t)=\sigma^2$ for every $t\in\widetilde\Gamma$, and for
$s,t\in\Gamma^\circ$
\[
r(s,t)=
\begin{cases}
\sigma^2 e^{-\kappa d(t,s)}\displaystyle\prod_{v\in P^\circ(t,s)}
\sqrt{p_{v,\bar v(t,s)}}, & \text{if } t\leadsto s,\\[1.2ex]
\sigma^2 e^{-\kappa d(s,t)}\displaystyle\prod_{v\in P^\circ(s,t)}
\sqrt{p_{v,\bar v(s,t)}}, & \text{if } s\leadsto t,\\[1.2ex]
0, & \text{if neither point is upstream of the other.}
\end{cases}
\]
The endpoint cases are obtained by the appropriate one-sided
edge limits.
\end{proposition}

So flow-connected points are correlated through the square root of the
accumulated flow proportions and flow-unconnected points are uncorrelated,
which is the defining shape of a tail-up model.  Taking $w_{v,\hat e}$ to be
the upstream drainage areas gives the usual weights.
In the exponential-covariance notation of \citet[p.~9]{Hoef10}, the parameter map is $\sigma^2=\theta_v$ and $\kappa=\theta_r^{-1}$.

\begin{proposition}[tail-down form]\label{prop:taildown}
Let $\Gamma$ be as in \Cref{prop:tailup} and let $\Gamma^{R}$ be $\Gamma$
with every edge reversed, with stationary anchoring at its single source.
Then every interior vertex of $\Gamma^R$ has one inflow, all three
conditions of \eqref{eq:vertex-condition} coincide with $\Bin\equiv1$, and
the field descends to $\Gamma$ with
$r(s,t)=\sigma^2 e^{-\kappa d(s,t)}$ for all $s,t\in\Gamma$.
\end{proposition}

Thus, varying the orientation and vertex condition in \eqref{eq:system}
gives the exponential $K_2$ tail-up and reversed-tree tail-down models, and under $C_V$, the construction is
related to the symmetric $\alpha=1$ Whittle--Mat\'ern field through the vertex
corrections in \Cref{prop:sym-vs-dir}. The coefficients $\Bin_v$ also permit
intermediate directed models whose weights may be estimated from data.

\section{Inference}
\label{sec:inference}
In this section we describe likelihood inference and prediction using the bridge
representation of \citet{part2jcgs}, adapted to the directed setting. Edge-local energies give block-diagonal endpoint precision matrices; imposing the vertex conditions gives the proper precision matrix. The
resulting sparse precision matrices avoid the dense covariance
factorizations associated with \Cref{sec:covariances}, making inference
efficient on large graphs.
Unlike the symmetric Whittle--Mat\'ern model, directed vertex conditions
need not identify all edge ends incident to the same geometric vertex. We
therefore retain the full split-edge endpoint vector rather than collapsing
it to one value per vertex. Further derivations and implementation details
are given in \Cref{supp:S4}.

\subsection{Preliminaries}
For matrices $\mv A_i$, let $\blkdiag(\mv A_i:i\in I)$ denote their
block-diagonal assembly in the stated order.
For each edge $e=[0,\ell_e]$, let $\widetilde u_e$ be an independent
boundaryless Whittle--Mat\'ern process with $\alpha=1$. For a field $f$, let $\gamma f:=\operatorname{col}_{e\in\mathcal E}\{f_e(0),f_e(\ell_e)\}$ be its vertically stacked endpoint traces in edge order. Set
$
\widetilde{\mv U}:=\gamma\widetilde u,$
and $\mv U:=\gamma u$.
For the symmetric (undirected)
model, \citet{bolin2023statistical} gives
\(\widetilde{\mv U}\sim \pN(\mv 0,\widetilde{\mv Q}^{-1})\), where
\(\widetilde{\mv Q}=\tau^2\blkdiag(\mv Q_e:e\in\mathcal E)\) and
\begin{align}
\label{eq:Qe_sym}
\mv Q_e
:=
\frac{\kappa}{e^{2\kappa\ell_e}-1}
\begin{bmatrix}
e^{2\kappa\ell_e}+1 & -2e^{\kappa\ell_e}\\[2pt]
-2e^{\kappa\ell_e} & e^{2\kappa\ell_e}+1
\end{bmatrix}.
\end{align}
The directed and symmetric edge forms differ only through their endpoint
traces. Consequently, they have the same zero-endpoint bridges, while the
directed construction replaces $\mv Q_e$ by the rank-one endpoint block
\begin{equation}
\label{eq:Qe_firstorder}
\mv Q^{L}_e
:=
\mv Q_e+\kappa
\begin{bmatrix}-1&0\\0&1\end{bmatrix}
=
\frac{2\kappa}{e^{2\kappa\ell_e}-1}
\begin{bmatrix}
1\\[2pt]
-e^{\kappa\ell_e}
\end{bmatrix}
\begin{bmatrix}
1&-e^{\kappa\ell_e}
\end{bmatrix}.
\end{equation}
This block is singular, with null vector
$(1,e^{-\kappa\ell_e})^\top$, corresponding to
$e^{-\kappa t}\in\ker L_e$.
Thus, the directed split-edge matrix represents an
intrinsic quadratic form, not the precision of a proper endpoint Gaussian.
The split-edge quadratic-form matrix is
\[
\mv Q=
\begin{cases}
\tau^2\blkdiag(\mv Q_e:e\in\mathcal E),
&\text{for the symmetric model},\\
\blkdiag\!\left(
\tau^2\mv Q_e^L+
\mathbf 1_{\{e\in\mathcal E_-\}}
\begin{bmatrix}
c_{\operatorname{tail}(e)}&0\\0&0
\end{bmatrix}:e\in\mathcal E
\right),
&\text{for the directed model}.
\end{cases}
\]
Thus source anchoring is added only to directed source-edge blocks.
For the likelihoods below, we use stationary anchoring,
$c_v=2\kappa\tau^2$.

To obtain the proper graph-wide model, impose the vertex conditions on the endpoint traces $\mv U$. For
${X\in\{C_V,K_1,K_2\}}$, write the finite-dimensional constraint defining
\(\widetilde H^1_X(\Gamma)\) as \(\mv K\mv U=\mv 0\), where \(\mv K\) has
rank \(k\). We next use this constrained representation to construct the
likelihood.

\subsection{Likelihood evaluation}
For observations at
\(s_i=(e_i,t_i)\in\widetilde\Gamma\), $i=1,\ldots,n$, consider
$Y_i=\mv x_i^\top\mv b+u(s_i)+\epsilon_i$, where \(\mv x_i\) contains the
covariates at \(s_i\), \(\mv b\) is the corresponding coefficient vector,
and the errors are independent of $u$ with
\(\epsilon_i\overset{\mathrm{iid}}{\sim}\pN(0,\sigma_\epsilon^2)\),
where $\sigma_\epsilon>0$.
For each edge \(e\), let \(\mv Y_e\), \(\mv X_e\), and \(\mv t_e\) collect
the observations, covariates, and edge coordinates, and define
\(\mv Z_e=\mv Y_e-\mv X_e\mv b\), with observed value $\mv z_e$.
Conditional on \(\mv U\), the edge blocks are independent and
${\mv Z_e\mid \mv U
    \sim
    \pN\left(\mv S_e(\mv t_e)\mv D_e\mv U,\boldsymbol{\Sigma}_e\right)}$, with 
    ${(\boldsymbol{\Sigma}_e)_{ij}
    =
    \sigma_\epsilon^2\mathbf 1(i=j)
    +
    r_{B,e}(t_{e,i},t_{e,j})}$.
Here \(\mv D_e\) is the endpoint map sending \(\mv U\) to
\((u_e(0),u_e(\ell_e))^\top\), \(\mv S_e(t)\) is the bridge interpolation matrix,
and \(r_{B,e}\) is the covariance function of the zero-endpoint bridge on
\(e\). Both $\mv S_e$ and
$r_{B,e}$ are the same for the directed and symmetric models because the
two edge forms differ only by a quadratic in the endpoint traces; see
\Cref{lem:edge-ips}.

Let \(\mathcal E_y\) be the set of edges with observations. Set
$\mv z=\operatorname{col}_{e\in\mathcal E_y}\mv z_e$,
$\mv B=\operatorname{col}_{e\in\mathcal E_y}
\{\mv S_e(\mv t_e)\mv D_e\}$, and
$\boldsymbol\Sigma=\blkdiag(\boldsymbol\Sigma_e:e\in\mathcal E_y)$.
Then \(\mv Z\mid \mv U=\mv u\sim
\pN(\mv B\mv u,\boldsymbol{\Sigma})\).  Following
\citet{bolin2021efficient}, we construct a vertex-local sparse row basis
$\mv T_c$ for $\ker(\mv K)$ and write
$\mv U=\mv T_c^\top\mv V$. 
Define $\mv Q_0=\mv T_c\mv Q\mv T_c^\top$ and
$\mv B_c=\mv B\mv T_c^\top$, and note that admissibility makes $\mv Q_0$ positive
definite. The posterior precision and mean are
$\mv Q_y=\mv Q_0+\mv B_c^\top\boldsymbol\Sigma^{-1}\mv B_c$ and
$\boldsymbol\mu_c=\mv Q_y^{-1}\mv B_c^\top\boldsymbol\Sigma^{-1}\mv z$.
Thus, with \(\boldsymbol{\theta} = \left(\sigma_\epsilon,\kappa,\tau \right)\), the log-likelihood is, up to an additive constant, 
\[
    2\ell(\boldsymbol{\theta},\mv b;\mv y)
    =
    \log|\mv Q_0|
    -
    \log|\mv Q_y|
    -
    \sum_{e\in\mathcal E_y}\log|\boldsymbol{\Sigma}_e|
    +
    \boldsymbol{\mu}_c^\top
    \mv Q_y
    \boldsymbol{\mu}_c
    -
    \mv z^\top\boldsymbol{\Sigma}^{-1}\mv z.
\]

\subsection{Prediction}

For fixed parameters, the posterior mean of the constrained endpoint vector is
$\widehat{\mv U}=\mv T_c^\top\boldsymbol\mu_c$.
For a prediction location
\(s^\star=(e^\star,t^\star)\in\widetilde\Gamma\), define the endpoint
component of the predictor by
${\widehat u_\Gamma(s^\star)=
\mv S_{e^\star}(t^\star)\mv D_{e^\star}\widehat{\mv U}}$. The conditional
mean is
\[
    \pE\{u(s^\star)\mid \mv Y=\mv y\}
    =
    \widehat u_\Gamma(s^\star)
    +
    r_{B,e^\star}(t^\star,\mv t_{e^\star})
    \boldsymbol{\Sigma}_{e^\star}^{-1}
    \left[
        \mv z_{e^\star}
        -
        \mv S_{e^\star}(\mv t_{e^\star})
        \mv D_{e^\star}\widehat{\mv U}
    \right],
\]
with the second term omitted if there are no observations on \(e^\star\).  This is the bridge kriging formula of \citet[Sec.~6.3]{part2jcgs} written for the directed
constraints.

\section{Applications}\label{sec:applications}
In this section we present two applications of the directed fields of
\Cref{sec:construction}.  The first is summer stream temperature on a
river network, where we compare the vertex conditions of
\Cref{sec:vertexconditions} against a symmetric reference and assess the
computational cost.  The second is traffic speeds on a road network, where
the orientation itself must be estimated from map data before the
directional models can be fitted. Both datasets are shown in Figure~\ref{fig:data}. In both cases the models are fitted by
maximum likelihood and compared through plug-in leave-one-out (LOO)
predictions, evaluated with the logarithmic score (LS) and the continuous
ranked probability score (CRPS) of \citet{gneiting2007strictly}, the
scaled CRPS (SCRPS) of \citet{bolin2023local}, and the mean absolute
(MAE) and root mean squared errors (RMSE).

\begin{figure}[H]
    \centering
    \includegraphics[width=0.55\linewidth,height=0.3\textheight]{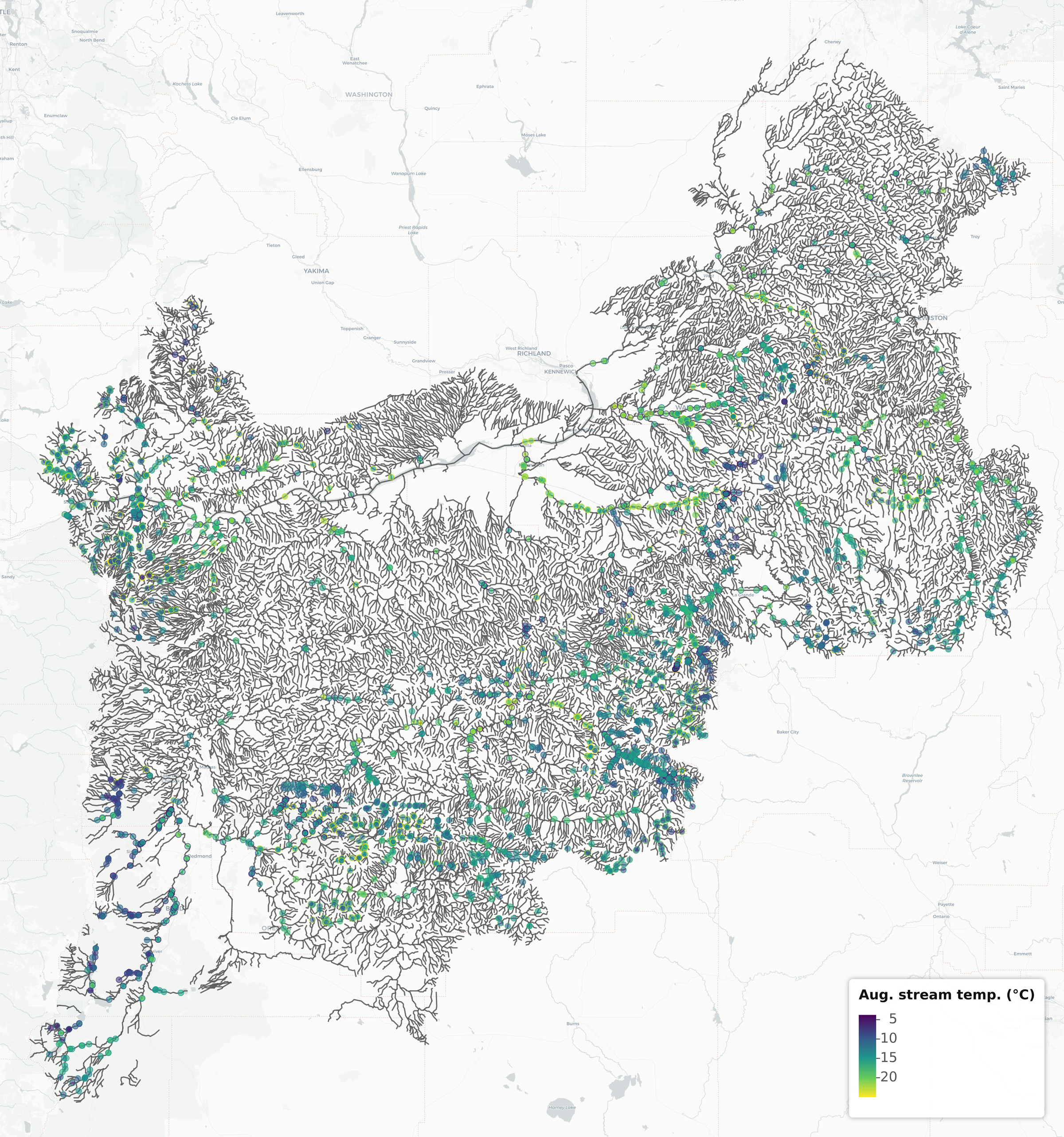}
    \includegraphics[width=0.4\linewidth,height=0.3\textheight]{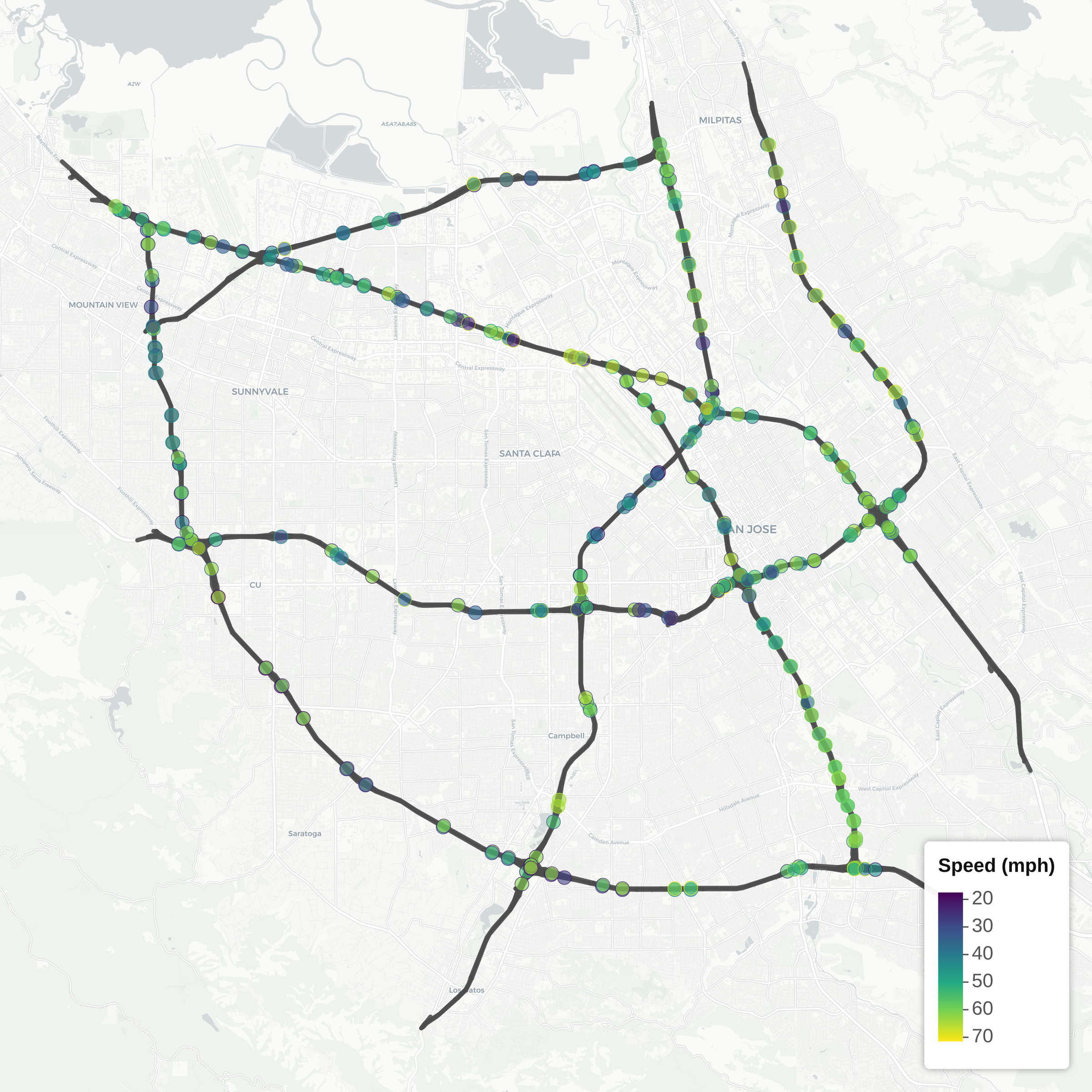}
    \caption{The Mid-Columbia River network (left) and the southern San Francisco Bay area traffic data (right).}
    \label{fig:data}
\end{figure}

\subsection{River network}
\label{subsec:river}

We use the Mid-Columbia River network, studied by
\citet{rivArt,isaak2017norwest}, to compare vertex conditions and their
computational cost.  The full data set contains 28\,613 edges and 2\,758 observation locations.
We limit ourselves to its largest connected component (18\,668 edges) and the corresponding $n=2\,080$ observations. All models include elevation,
slope, and precipitation as covariates, independent measurement error, and
one latent graph field.  

The $K_1$ and $K_2$ fields use drainage-area weights. Reversing the river orientation and applying
$K_1$ gives tail-down continuity: each reversed confluence has one inflow, so its
normalized weight is one, and flow-unconnected branches that share downstream
flow may be dependent in the original orientation.

\begin{table}[t]
\centering
\small
\setlength{\tabcolsep}{2.5pt}
\caption{Fitted log-likelihood, optimization time, and plug-in LOO scores for
the Mid-Columbia data. Lower is better and bold marks the best performance.}
\label{tab:midcolumbia-loo}
\begin{tabular}{@{}lrrrrrrr@{}}
\toprule
& \multicolumn{5}{c}{Plug-in LOO score} & \multicolumn{2}{c}{Fit} \\
\cmidrule(lr){2-6}\cmidrule(l){7-8}
Model & LS & CRPS & SCRPS & MAE & RMSE & Log lik. & Opt. (s) \\
\midrule
Tail-up, linear ($K_1$)
& \textbf{1.8993} & \textbf{0.8868} & \textbf{1.2759}
& \textbf{1.2033} & 1.6766 & -4265.53 & 12.17 \\
Tail-up, square-root ($K_2$)
& 1.9002 & 0.8869 & 1.2772 & 1.2119
& \textbf{1.6731} & \textbf{-4248.22} & 12.64 \\
Symmetric WM ($\alpha=1$)
& 1.9922 & 0.9716 & 1.3272 & 1.3271
& 1.8028 & -4433.07 & 20.89 \\
Tail-down, continuity
& 1.9862 & 0.9642 & 1.3238 & 1.3178
& 1.7928 & -4408.65 & 8.97 \\
\bottomrule
\end{tabular}

\end{table}

Table~\ref{tab:midcolumbia-loo} reports all plug-in LOO scores, 
the fitted log-likelihood, and optimization time.
The two tail-up fits perform nearly identically and substantially better than the other models. $K_1$ has slightly lower LS, CRPS, SCRPS, and MAE, whereas $K_2$ has the larger log-likelihood and slightly lower RMSE. The performance difference between the two tail-up models is negligible
relative to their advantage over the symmetric and tail-down models.
All optimizations finish within 21 seconds. A separate computational scaling
experiment comparing the sparse profile likelihood with direct dense-covariance
evaluation is reported in \Cref{app:computational-scaling}.

\subsection{Traffic data}
The \texttt{pems} data in \texttt{MetricGraph} contain $n=325$ traffic-speed
observations $y_i$, measured in miles per hour at locations $s_i$ on an
$848$-edge road network in the southern San Francisco Bay area. Directional
modelling requires the edge orientation to follow traffic flow, whereas the
stored orientations are inherited from the ordering of the source lines. We
therefore match each edge to its nearest OpenStreetMap (OSM) way. OSM one-way
metadata reverse $88$ edges, and propagation along degree-two chains reverses
a further $82$. Thus the OSM-oriented graph differs from the raw graph on
$170$ edges, or $20\%$ of the network.

Following \citet{bolin2023statistical}, we use five non-directional baselines:
a resistance-metric isotropic exponential field (isoExp) \citep{anderes2020isotropic}, vertex-indexed graph-Laplacian Mat\'ern fields
$\mathrm{GL}_\alpha$ \citep{BorovitskiyEtAl2021}, and Whittle--Mat\'ern fields $\mathrm{WM}_\alpha$, with
$\alpha\in\{1,2\}$ for the latter two classes. The directional candidates use
the conditions $K_1$ and $K_2$ with edge-specific weights
$w_e$. We fit both conditions with unit weights and on the OSM-oriented graph, we also fit $K_1$ using
$w_e=\ell_e$, threshold weights $w_e=5$ for $\ell_e>0.5\,\mathrm{km}$ and
$w_e=1$ otherwise, and OSM road-class weights $8\!:\!6\!:\!4\!:\!2\!:\!1$
for motorway, trunk, primary, secondary, and all remaining classes. Only
these ratios matter because the vertex conditions normalize the incoming
weights.

\begin{table}[t]
\centering
\small
\setlength{\tabcolsep}{2.5pt}
\caption{Plug-in LOO scores and negative log-likelihood (NLL), $n=325$.
The blocks contain non-directional baselines, directional fits on the raw
graph, and directional fits on the OSM-oriented graph. Lower values are
better; bold marks the column minima.}
\label{tab:loo}
\begin{tabular}{@{}lrrrrrr@{}}
\toprule
Model & LS & CRPS & SCRPS & MAE & RMSE & NLL \\
\midrule
$\mathrm{isoExp}$ & 3.604 & 4.735 & 2.133 & 6.175 & 8.610 & 1223.84 \\
$\mathrm{GL}_{\alpha=1}$ & 3.605 & 4.737 & 2.133 & 6.175 & 8.611 & 1221.38 \\
$\mathrm{GL}_{\alpha=2}$ & 3.551 & 4.531 & 2.107 & 5.853 & 8.304 & 1208.70 \\
$\mathrm{WM}_{\alpha=1}$ & 3.605 & 4.738 & 2.133 & 6.178 & 8.614 & 1221.23 \\
$\mathrm{WM}_{\alpha=2}$ & 3.549 & 4.522 & 2.107
& \textbf{5.841} & \textbf{8.256} & 1208.00 \\
\midrule
$K_1$, raw, unit & 3.571 & 4.652 & 2.116 & 6.174 & 8.608 & 1202.86 \\
$K_2$, raw, unit & 3.569 & 4.669 & 2.117 & 6.246 & 8.598 & 1202.59 \\
\midrule
$K_1$, OSM, unit & 3.562 & 4.641 & 2.114 & 6.188 & 8.550 & 1204.20 \\
$K_2$, OSM, unit & 3.559 & 4.625 & 2.114 & 6.240 & 8.463 & 1204.49 \\
$K_1$, OSM, length & 3.571 & 4.655 & 2.118 & 6.185 & 8.553 & 1207.94 \\
$K_1$, OSM, threshold & 3.551 & 4.553 & 2.107 & 6.056 & 8.359 & 1203.30 \\
$K_1$, OSM, road class & \textbf{3.541} & \textbf{4.513}
& \textbf{2.102} & 5.979 & 8.335 & \textbf{1199.72} \\
\bottomrule
\end{tabular}
\end{table}

We fit all models by maximum likelihood and evaluate plug-in LOO predictions using the same scoring rules as before.
\Cref{tab:loo} shows that correcting the orientation modestly improves the
unit-weight directional fits in most scores, whereas the choice of edge weights has a larger
effect. The road-class $K_1$ model gives the smallest LS, CRPS, SCRPS, and
NLL. The non-directional $\mathrm{WM}_{\alpha=2}$ model instead minimizes MAE
and RMSE.
Thus the directional model with road-class weights improves distributional
prediction and likelihood fit, while the smoother non-directional model
retains a small advantage for point prediction.

\section{Discussion}
\label{sec:discussion}

We have shown that a single first-order system, \eqref{eq:system}, produces
a family of Gaussian fields on a directed metric graph whose members are
selected by two choices: the vertex condition and the orientation.
The construction is exact and mesh-free, admits
sparse-precision inference, and contains the tail-up and tail-down models of
\Cref{sec:tailupdown} as special cases, alongside new directed models with
physically motivated vertex conditions. Under continuity, its Cameron--Martin form differs from the symmetric $\alpha=1$ Whittle--Mat\'ern form by explicit vertex and source-anchor terms.

A potential limitation with directed models is that the framework presumes that each edge carries a direction.  On a
river network this is given by the hydrology, but on a road network may be more difficult to obtain, and a two-way road is only crudely represented by a single directed edge. Handling such roads would require either a pair of opposed edges or a vertex condition that mixes the two directions. 

As the operator $\kappa+\partial_\Gamma$ is of first order, 
the fields have the regularity of the $\alpha=1$ Whittle--Mat\'ern field and no smoother member of the family is available. 
Obtaining directed analogues of $\alpha>1$ is thus a natural next step.
%
A natural extension is spatio-temporal modelling: the directed field describes network transport and a temporal operator describes time evolution.

\section*{Acknowledgements}

We thank S{\o}ren Wengel Mogensen for pointing out the connection between
our construction and the trek rule.

\section*{Software and data availability}

All models are implemented within the \texttt{MetricGraph} \textsf{R} package.

\bibliographystyle{abbrvnat}
\bibliography{ref_jw}



\newpage
\appendix

\section{Cyclic admissibility: expanded algebra}\label{supp:S2}

This section expands the calculation behind the cyclic counterexample retained
in the main-text admissibility section.  Order the edge-tail amplitudes as
$\eta=(\eta_1,\eta_2,\eta_3,\eta_4)^\top$.  The source edge $e_1$ has no
feedback row.  At $v$, the incoming weights of $(e_1,e_4)$ are $(1,3)$, so the
$K_2$ coefficients are $(1/2,\sqrt{3}/2)$; at $w$, the equal incoming weights
of $(e_2,e_3)$ give coefficients $(1/\sqrt{2},1/\sqrt{2})$.  With
$a_j=e^{-\kappa\ell_j}$, the transfer matrix is
\[
\Mtr=
\begin{bmatrix}
0&0&0&0\\
a_1/2&0&0&\sqrt3\,a_4/2\\
a_1/2&0&0&\sqrt3\,a_4/2\\
0&a_2/\sqrt2&a_3/\sqrt2&0
\end{bmatrix}.
\]
For a zero-energy mode the anchored source forces $\eta_1=0$, and
$(I-\Mtr)\eta=0$ becomes
\[
\eta_2=\eta_3=\frac{\sqrt3}{2}a_4\eta_4,\qquad
\eta_4=\frac{a_2\eta_2+a_3\eta_3}{\sqrt2}.
\]
Consequently a nonzero mode exists precisely when
\[
1=\frac{\sqrt3}{2\sqrt2}\,a_4(a_2+a_3).
\]
Taking $\ell_2=\ell_3$ reduces this to
$1=\sqrt{3/2}\exp\{-\kappa(\ell_2+\ell_4)\}$, from which
${\ell_2+\ell_4=(2\kappa)^{-1}\log(3/2)}$.  Choosing $\eta_4=1$ then gives the
null mode displayed in the article.

This calculation isolates the issue: the square-root coefficients are
variance preserving at a confluence on a tree, but their row sum can exceed
one.  Feedback around a directed cycle can therefore produce unit loop gain.

\section{Auxiliary covariance notation and results}\label{supp:S3}

The formal equivalence between the acyclic forward system and a proper global
OU construction is stated in \Cref{lem:graph-OU-form}, and its forward
substitution consequence is \Cref{cor:forward-substitution} of the main
article. We use that notation below to record auxiliary covariance formulas
that are omitted from the main text.

Retain the notation $Z_j$, $\mathcal J(y)$ and $\Lambda^\uparrow(y)$ from
\Cref{sec:covariances}. For $H\subset\Gamma$ a finite union of edge intervals,
set $\mathcal I_e(H):=\{\xi:(e,\xi)\in\pi^{-1}(H)\}$ and
$\mathcal E(H):=\{e:\mathcal I_e(H)\ne\varnothing\}$. For $g$ deterministic
on $\pi^{-1}(H)$, interpret stochastic and ordinary integrals edgewise by
\[
\int\limits_H g(z)\,dW(z)
:= \!\!
\sum_{e\in\mathcal E(H)}\int\limits_{\mathcal I_e(H)} \!\!g((e,\xi))\,dW_e(\xi),
\quad
\int_H g(z)\,dz
:=\!\!
\sum_{e\in\mathcal E(H)}\int\limits_{\mathcal I_e(H)} \!\!g((e,\xi))\,d\xi.
\]

Taking the two evaluation points equal in the covariance theorem yields the
following formulas.
\begin{corollary}\label{cor:variance-formula}
Under the assumption of \Cref{thm:cov-kernel}, for every
$s\in\Gamma^\circ$
\[
r(s,s)=\Var(u(s))
=
\sum_{j\in\mathcal J(s)}\sigma^2_{s_{-,j}}\,\mathsf A(s_{-,j},s)^2
+\tau^{-2}\int_{\Lambda^\uparrow(s)} \mathsf A(z,s)^2\,dz,
\]
and if $x,y\in\Gamma^\circ$ satisfy $x\leadsto y$ on a common edge then
$\mathsf A(x,y)=e^{-\kappa d(x,y)}$ and
\begin{equation}
\label{eq:within-edge-variance}
r(y,y)=e^{-2\kappa d(x,y)}\,r(x,x)+\frac{1-e^{-2\kappa d(x,y)}}{2\kappa\tau^2}.
\end{equation}
\end{corollary}

\Cref{supp:S6} gives the transfer representation and proofs.  These statements
are restricted to acyclic directed graphs; they are not used for the cyclic
traffic likelihood.

\section{Endpoint precision, constraints, likelihood, and prediction}\label{supp:S4}

\subsection{Endpoint precision and the bridge identity}

The article gives the symmetric endpoint block $\mv Q_e$ in
\eqref{eq:Qe_sym} and its directed counterpart $\mv Q_e^L$ in
\eqref{eq:Qe_firstorder}.  Here is the general boundary-update argument behind
that relation.  By \Cref{lem:edge-ips},
\[
a_e^L(f,g)
=a_e^A(f,g)
+\kappa\{f(\ell_e)g(\ell_e)-f(0)g(0)\},
\]
so the two forms differ only by a quadratic in the endpoint traces.

\begin{lemma}\label{lem:boundary-update-general}
Let $H\subset H^1(e)$ be a linear space on the edge $e=[0,\ell_e]$ and let
$\gamma_e f=(f(0),f(\ell_e))^\top$.  Set
$S=\mathrm{Ran}(\gamma_e)$, and let $\mv P_S$ be the orthogonal projector
onto $S$.  Suppose two symmetric positive semidefinite bilinear forms
$a_1,a_2$ on $H$ satisfy
$a_2(f,g)=a_1(f,g)+(\gamma_e f)^\top\mv G\,(\gamma_e g)$ for all
$f,g\in H$, with $\mv G\in\mathbb R^{2\times2}$ symmetric.  Define
\[
I_i(\mv z):=\inf\{\tfrac12a_i(f,f):\ f\in H,\ \gamma_e f=\mv z\},
\qquad \mv z\in S,\ i=1,2,
\]
and let $\mv Q_i$ be the symmetric matrix representing $I_i$, in that
$I_i(\mv z)=\tfrac12\mv z^\top\mv Q_i\mv z$ for $\mv z\in S$ and
${\mv Q_i=\mv P_S\mv Q_i\mv P_S}$.  Then
\[
I_2(\mv z)=I_1(\mv z)+\tfrac12\mv z^\top\mv G\mv z,
\qquad
\mv Q_2=\mv Q_1+\mv P_S\mv G\mv P_S.
\]
\end{lemma}
\begin{proof}
Fix $\mv z\in S$ and set
$
\mathcal F_{\mv z}:=\{f\in H:\gamma_e f=\mv z\}.
$
This set is nonempty because $S=\mathrm{Ran}(\gamma_e)$. For every
$f\in\mathcal F_{\mv z}$, the assumed relation between the bilinear forms
gives
\[
\frac12 a_2(f,f)
=
\frac12 a_1(f,f)+\frac12\mv z^\top\mv G\mv z.
\]
The second term depends only on $\mv z$, and hence is constant on
$\mathcal F_{\mv z}$. Therefore,
\[
\begin{aligned}
I_2(\mv z)
&=\inf_{f\in\mathcal F_{\mv z}}
  \left\{\frac12a_1(f,f)+\frac12\mv z^\top\mv G\mv z\right\} \\
&=I_1(\mv z)+\frac12\mv z^\top\mv G\mv z.
\end{aligned}
\]

Since $\mv z=\mv P_S\mv z$ for $\mv z\in S$, we also have
\[
\mv z^\top\mv G\mv z
=
\mv z^\top\mv P_S\mv G\mv P_S\mv z.
\]
Thus, for every $\mv z\in S$,
\[
I_2(\mv z)
=
\frac12\mv z^\top
\bigl(\mv Q_1+\mv P_S\mv G\mv P_S\bigr)\mv z.
\]
The matrix in parentheses is symmetric and is supported on $S$, since
\[
\mv P_S
\bigl(\mv Q_1+\mv P_S\mv G\mv P_S\bigr)
\mv P_S
=
\mv Q_1+\mv P_S\mv G\mv P_S.
\]
It is therefore the canonical symmetric representative of $I_2$, and
hence
$
\mv Q_2=\mv Q_1+\mv P_S\mv G\mv P_S.
$
\end{proof}

For $H=H^1(e)$ the trace map is onto, so $S=\mathbb R^2$ and
$\mv P_S=\mv I$.  Taking
$\mv G=\kappa\operatorname{diag}(-1,1)$ proves
$\mv Q_e^L=\mv Q_e+\kappa\operatorname{diag}(-1,1)$, as stated in
\eqref{eq:Qe_firstorder}.  Moreover, the update is constant after conditioning
on the two endpoints and vanishes on the zero-endpoint subspace.  This proves
that the directed and symmetric constructions have the same interpolation
matrix $\mv S_e$ and bridge covariance $r_{B,e}$.

\subsection{Matrix representation of the vertex conditions}\label{app:matrix-bc}

The vertex conditions of \Cref{sec:vertexconditions} admit a convenient finite-dimensional matrix representation. Using the endpoints
$\widetilde{\mathcal V}_v$ and their in--out decomposition from
\Cref{sec:graphs}, fix an ordering
${\widetilde{\mathcal V}_v=\{\xi_{v,1},\dots,\xi_{v,d_v}\}}$, where $d_v=\deg(v)$, and define the local boundary trace operator mapping ${\gamma_v: \widetilde H^1(\Gamma)\to\mathbb R^{d_v}}$ by
$\gamma_v f=\bigl(f(\xi_{v,1}),\dots,f(\xi_{v,d_v})\bigr)^\top$.
Thus, $\gamma_v f$ collects the endpoint values of $f$ at all edge ends incident to $v$.

Writing $\gamma f\in\mathbb R^{2|\mathcal E|}$ for the global trace operator of $f$, ordered by edge, and
$\gamma_v f\in\mathbb R^{\deg(v)}$ for its restriction to the
endpoints in $\widetilde{\mathcal V}_v$, any homogeneous local linear
condition on the endpoint values can be written as $\mv K_v\gamma_v f=0$. Collecting these vertexwise conditions gives the global representation $\mv K\gamma f=0$, where $\mv K$ is obtained by placing the columns of the blocks $\mv K_v$ according to the chosen global endpoint ordering.

This matrix formulation of endpoint constraints corresponds to the construction in quantum-graph theory. There, the framework is typically developed for second-order differential operators, for which the vertex conditions must account not only for the edgewise function values but also for their outward derivatives. Accordingly, introducing the outward derivative
trace $\gamma_v^\partial f\in\mathbb R^{\deg(v)}$, the general
homogeneous local linear condition takes the form $\mathbf K_v\gamma_v f+\mathbf D_v\gamma_v^{\partial}f=0$, where $\mv D_v$ is the coefficient matrix specifying, for each local constraint, the linear combination of outward derivative traces entering that constraint
\cite[Section~1.4.1]{BerkolaikoKuchment2013} and
\cite[Section~3.1, Eq.~(5), Theorem~3]{Kuchment2004QuantumGraphsI}.

We now record the concrete matrices corresponding to the three conditions of \eqref{eq:vertex-condition}. Continuity, $C_V$, can be written at each vertex as $\mv K_v^{C}\gamma_v f=0$, where
\[
\mv K_v^{C}
=
\begin{bmatrix}
1 & -1 & 0 & \cdots & 0\\
0 & 1 & -1 & \cdots & 0\\
\vdots & \ddots & \ddots & \ddots & \vdots\\
0 & \cdots & 0 & 1 & -1
\end{bmatrix}
\in\mathbb R^{(d_v-1)\times d_v}.
\]
Next, order \(\widetilde{\mathcal V}_v\) with the incoming endpoints first and
the outgoing endpoints second, and write
$
\gamma_v f=
\Bigl(
\bigl(f(e_i^{\mathrm{in}},\ell_{e_i^{\mathrm{in}}})\bigr)_{i=1}^{m_v},
\bigl(f(e_j^{\mathrm{out}},0)\bigr)_{j=1}^{n_v}
\Bigr)^\top,
$
where \(m_v=|\mathcal E_v^{\mathrm{in}}|\) and
\(n_v=|\mathcal E_v^{\mathrm{out}}|\). In the notation of
\eqref{eq:vertex-condition}, set
$\beta^1_{v,i}=p_{v,e_i^{\mathrm{in}}}$ and 
$\beta^2_{v,i}=\sqrt{p_{v,e_i^{\mathrm{in}}}}$,
so that $\beta^k_{v,i}=\Bin_v(e,e_i^{\mathrm{in}})$ under $K_k$, for $k=1,2$ and
any $e\in\mathcal E_v^{\mathrm{out}}$, and write
$\beta^1_v=(\beta^1_{v,1},\dots,\beta^1_{v,m_v})^\top$ and
$\beta^2_v=(\beta^2_{v,1},\dots,\beta^2_{v,m_v})^\top$.  Then $K_1$ and $K_2$ are respectively equivalent to $\mv K_v^{K_1}\gamma_v f=0$ and $\mv K_v^{K_2}\gamma_v f=0$, where
\[
\mv K_v^{K_1}
=
\begin{bmatrix}
-\mathbf 1_{n_v}(\beta^1_v)^\top & I_{n_v}
\end{bmatrix},
\qquad
\mv K_v^{K_2}
=
\begin{bmatrix}
-\mathbf 1_{n_v}(\beta^2_v)^\top & I_{n_v}
\end{bmatrix}.
\]
If $n_v=0$, the $K_1$ and $K_2$ blocks have no rows,
whereas $\mv K_v^C$ still equates the incoming traces.
Hence the three spaces $\widetilde H^1_X(\Gamma)$ of \Cref{sec:vertexconditions}
admit the unified description
\[
\widetilde H^1_X(\Gamma)
=
\Bigl\{
f\in\widetilde H^1(\Gamma):
\mv K^{X}\gamma f=0
\Bigr\},
\qquad
X\in\{C_V,K_1,K_2\},
\]
where $\mv K^{X}$ is obtained by assembling the corresponding local blocks over all $v\notin\mathcal V_-$.

\begin{example}
We illustrate the local matrix formulation at a vertex with two incoming edges and one outgoing edge, namely $v_2$ in the orientation $\Gamma_{\mathrm{in}}$ of \Cref{fig:inout}(a), whose vertex-split representation is \Cref{fig:inout}(c).
At vertex $v_2$ we have $\mathcal E_{v_2}^{\mathrm{in}}=\{e_2,e_3\}$ and $\mathcal E_{v_2}^{\mathrm{out}}=\{e_1\}$, so that $m_{v_2}=2$ and $n_{v_2}=1$. Ordering the incoming endpoints first, the trace is $\gamma_{v_2}f=\bigl(f(e_2,\ell_{e_2}),\,f(e_3,\ell_{e_3}),\,f(e_1,0)\bigr)^\top$. The continuity condition $C_V$ is then equivalent to $\mv K_{v_2}^{C}\gamma_{v_2}f=0$, where
\[
\mv K_{v_2}^{C}
=
\begin{bmatrix}
1 & -1 & 0\\
0 & 1 & -1
\end{bmatrix}.
\]
For $K_1$, the condition is $\mv K_{v_2}^{K_1}\gamma_{v_2}f=0$, where
\[
\mv K_{v_2}^{K_1}
=
\begin{bmatrix}
-\beta_{v_2,1}^1 & -\beta_{v_2,2}^1 & 1
\end{bmatrix}
=
\begin{bmatrix}
-p_{v_2,e_2} & -p_{v_2,e_3} & 1
\end{bmatrix},
\]
with $p_{v_2,e_i}=w_{v_2,e_i}/(w_{v_2,e_2}+w_{v_2,e_3})$ for $i=2,3$.  The corresponding matrix for $K_2$ is obtained by replacing $\beta_{v_2}^1$ with $\beta_{v_2}^2$.
\end{example}

\subsection{Constraint complement and Gaussian integration}\label{supp:likelihood}

Let $\mv K$ be the assembled full-row-rank constraint matrix and let
$\mv T_c$ be the sparse row basis used by the implementation for
$\ker(\mv K)$.  Thus $\mv K\mv T_c^\top=0$ and every admissible endpoint
vector has the unique form $\mv U=\mv T_c^\top\mv V$.  If $\mv Q$ is the
block-diagonal split-edge quadratic-form matrix, then the prior precision of $\mv V$ is
$\mv Q_0=\mv T_c\mv Q\mv T_c^\top$. More explicitly, if
$m=2|\mathcal E|$ and $\mv K\in\mathbb R^{k\times m}$, then
$\mv T_c\in\mathbb R^{(m-k)\times m}$,
$\mv V\in\mathbb R^{m-k}$, and
$\mv Q_0\in\mathbb R^{(m-k)\times(m-k)}$. Because $\mv K$ is assembled from
vertex-local blocks and $\mv Q$ from edge-local blocks, $\mv T_c$ and
$\mv Q_0$ retain the graph sparsity.

Assume throughout this subsection that the constrained prior precision
$\mv Q_0$ is positive definite and that $\sigma_\epsilon>0$. Then every
$\boldsymbol\Sigma_e$, and hence $\boldsymbol\Sigma$, is positive definite.

Stacking the edgewise bridge regressions gives
$\mv Z\mid\mv V\sim \pN(\mv B_c\mv V,\boldsymbol\Sigma)$, where
$\mv B_c=\mv B\mv T_c^\top$. Here
$\mv B_c\in\mathbb R^{n\times(m-k)}$, while the block-diagonal
$\boldsymbol\Sigma$ keeps the observation contribution edge-local.
Completing the square gives
\[
\mv Q_y=\mv Q_0+\mv B_c^\top\boldsymbol\Sigma^{-1}\mv B_c,\qquad
\boldsymbol\mu_c=\mv Q_y^{-1}\mv B_c^\top
 \boldsymbol\Sigma^{-1}\mv z.
\]
The identity
\[
\mv v^\top\mv Q_0\mv v+
(\mv z-\mv B_c\mv v)^\top\boldsymbol\Sigma^{-1}
(\mv z-\mv B_c\mv v)
=(\mv v-\boldsymbol\mu_c)^\top\mv Q_y
(\mv v-\boldsymbol\mu_c)
+\mv z^\top\boldsymbol\Sigma^{-1}\mv z
-\boldsymbol\mu_c^\top\mv Q_y\boldsymbol\mu_c,
\]
and Gaussian integration yield the determinant expression in the main
article.  

For a prediction point $s^\star=(e^\star,t^\star)$, conditioning first on
$\mv U$ separates the zero-endpoint bridge on $e^\star$ from all other edge
bridges.  Taking the posterior expectation of
$\mv U=\mv T_c^\top\mv V$ gives
$\widehat{\mv U}=\mv T_c^\top\boldsymbol\mu_c$; ordinary Gaussian
conditioning of the local bridge then gives the predictor displayed in the
article. Adding $\mv x(s^\star)^\top\mv b$ gives the corresponding
response predictor. No additional bridge result is required for either
derivation.

\section{Application construction and scaling experiment}\label{supp:S5}

\subsection{River-network preprocessing}

Drainage area supplies the normalized incoming weights in the Mid-Columbia
analysis.  It is used
as a proxy for discharge, so the $K_1$ analysis should be read as an
idealized complete-mixing model rather than as a measured heat-balance model.
The timings in \Cref{tab:midcolumbia-loo} cover numerical maximization only;
they omit data preparation, fixed-effect recovery and plug-in LOO prediction.

\subsection{Road-network admissibility check}

Because the road network is cyclic, its fitted directional models use the
constrained endpoint likelihood rather than the acyclic trek covariance.  For
every fit, a successful sparse Cholesky factorisation of the restricted
precision confirms numerical positive definiteness, and hence admissibility to
the factorisation tolerance, at the fitted parameter values.

\subsection{Computational scaling experiment}
\label{app:computational-scaling}

To assess likelihood-evaluation scaling, we used the pruned largest connected
component of the Mid-Columbia network and sampled observation locations
uniformly across its edges for ten values of $n$ from $100$ to $20\,000$. For
$K_1$, $K_2$, and continuity on the reversed graph,
\Cref{fig:sim-speed-evaluate} reports the median of five timings at a fixed
parameter vector, comparing the precision-based likelihood of
\Cref{sec:inference} with direct covariance-based evaluation. The $K_1$ and
$K_2$ covariance matrices are also structurally sparse because observations
on flow-unconnected branches have zero covariance, although the implementation
constructs the full $n\times n$ matrix before sparse factorisation. The three
precision-based evaluations take $1.18$--$1.22$ seconds at $n=20\,000$.
Covariance-based $K_1$ and $K_2$ take approximately $19$--$21$ seconds at
$n=12\,000$, while covariance-based continuity takes $71$ seconds at
$n=8\,000$; larger cases were not attempted because of memory requirements.
These wall-clock times are descriptive of the benchmark run and should not
be interpreted as hardware-independent performance guarantees.

\begin{figure}[H]
    \centering
    \includegraphics[width=0.72\linewidth]{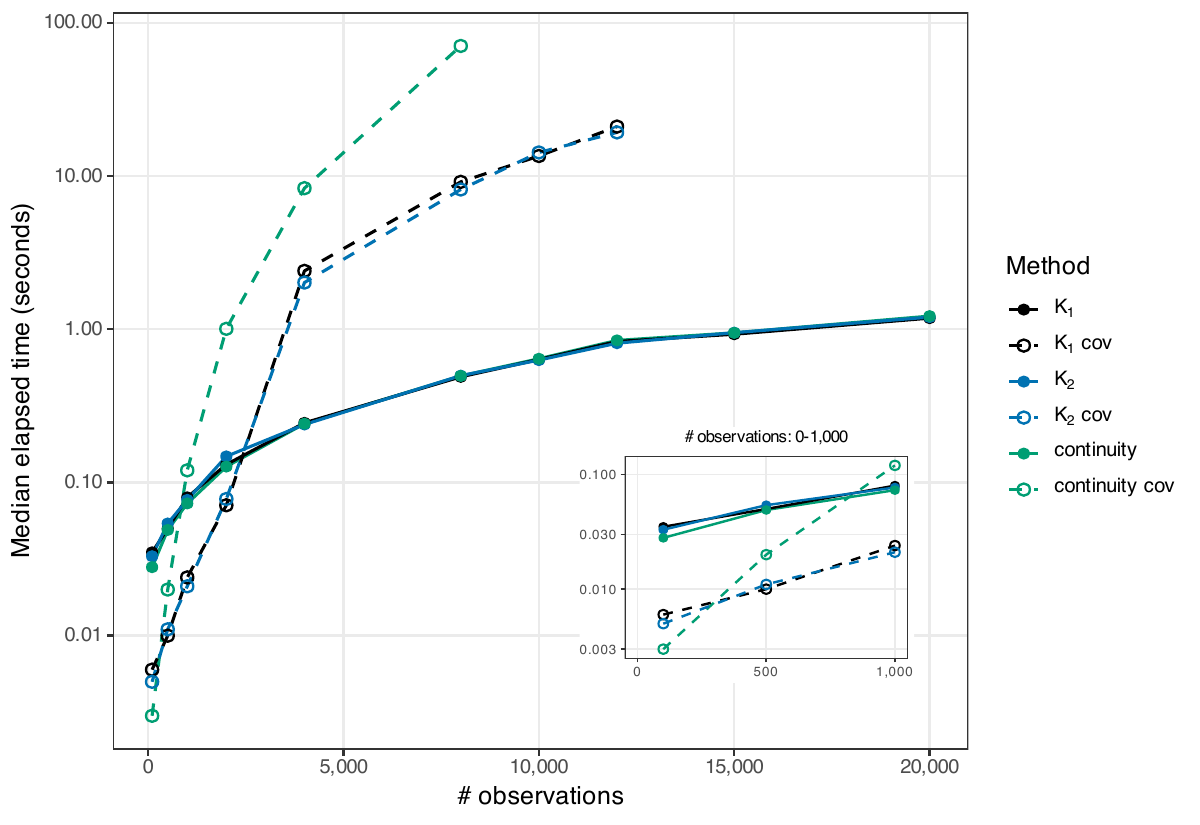}
    \caption{Median likelihood-evaluation time over five runs. Solid lines
    denote the precision-based likelihood and dashed lines direct
    covariance-based evaluation; the inset shows $n\leq 1\,000$.}
    \label{fig:sim-speed-evaluate}
\end{figure}

\section{Proofs}\label{supp:S6}

\subsection{Proofs for Section~\ref{sec:construction}}
\begin{proof}[Proof of \Cref{lem:edge-ips}]
Expanding $(L_ef)(L_eg)=\kappa^2fg+\kappa(fg)'+f'g'$ and integrating the
middle term gives
$a_e^L(f,g)=a^A_e(f,g)
+\kappa\{f(\ell_e)g(\ell_e)-f(0)g(0)\}$, which is the claim.
\end{proof}

\begin{proof}[Proof of \Cref{lem:innerProdAndRKHS}]
Write $h_0(t):=e^{-\kappa t}$ and split the solution \eqref{eq:OU} as
$u=u_0h_0+u^\circ$, where
${u^\circ(t):=\tau^{-1}\int_0^t e^{-\kappa(t-s)}\,dW_e(s)}$ is the
noise-input component on $e$, and $u_0h_0$ is the initial-value component.
Define the Volterra operator
\[
    (Gh)(t):=\int_0^t e^{-\kappa(t-s)}h(s)\,ds.
\]
Brownian motion has Cameron--Martin space
$H_0^1(e):=\{h\in H^1(e):h(0)=0\}$,
identified isometrically with $L_2(e)$ by differentiation. Thus a control
$h\in L_2(e)$ generates the Cameron--Martin path $\tau^{-1}Gh$ for
$u^\circ$. By
Theorem~3C of \cite{parzen1961regression}, $u^\circ$ has RKHS
$\operatorname{Ran}(\tau^{-1}G)$, that is
    $\mathcal H_{u^\circ}=\{f\in H^1(e):f(0)=0\}$,
    and 
    $\langle f,g\rangle_{\mathcal H_{u^\circ}}
    =\tau^{2}\langle L_ef,L_eg\rangle_{L_2(e)}$,
since $G$ inverts $L_e$ on functions vanishing at $0$.

The initial term has RKHS
\(\mathcal H_0=\operatorname{span}\{h_0\}\), with
\(\langle ah_0,bh_0\rangle_{\mathcal H_0}=\sigma_0^{-2}ab\).
The initial-value and noise-input Gaussian summands are independent.
As random elements of \(C(e)\), the two terms have supports
\(\{f\in C(e):f(0)=0\}\) and \(\operatorname{span}\{h_0\}\).
These supports intersect only at zero, and the latter is one-dimensional
and therefore complemented, so
Lemma~9.1 of \citet{vanderVaartVanZanten2008} gives the orthogonal
sum \(\mathcal H_u=\mathcal H_{u^\circ}\oplus\mathcal H_0\).

Every \(f\in H^1(e)\) has the unique decomposition
$
    f=\{f-f(0)h_0\}+f(0)h_0,
$
and \(L_eh_0=0\), so adding the two component inner products gives
\[
    \langle f,g\rangle_{\mathcal H_u}
    =
    \tau^{2}\langle L_e f,L_e g\rangle_{L_2(e)}
    +\sigma_0^{-2}f(0)g(0)
    =
    \tau^2a_e^L(f,g)+\sigma_0^{-2}f(0)g(0),
\]
where the last step is by \eqref{eq:QuadForm}.  
\end{proof}

\subsection{Proofs for Section~\ref{sec:rkhs}}
\begin{proof}[Proof of \Cref{lem:q_is_inner_product}]
The form \eqref{eq:q_semi} is bilinear and symmetric on
$\widetilde H^1(\Gamma)$, and
\[
a_\Gamma^0(f,f)
=\sum_{v\in\mathcal V_-}c_v\,f(v)^2
+\tau^2\sum_{e\in\mathcal E}a_e^L(f_e,f_e)
\ \ge\ 0,
\]
so it is a positive semidefinite bilinear form on any
subspace. It defines an inner product on $\widetilde H^1_X(\Gamma)$ exactly
when it is definite there, i.e., when
$a_\Gamma^0(f,f)=0\Longrightarrow f=0$ for
$f\in\widetilde H^1_X(\Gamma)$,
and the whole proof consists of determining when
$\widetilde H^1_X(\Gamma)$ contains a nonzero element of zero energy.

Suppose then that $f\in\widetilde H^1_X(\Gamma)$ has
$a_\Gamma^0(f,f)=0$.
Since the $c_v$ are strictly positive, $f(v)=0$ for every
$v\in\mathcal V_-$, and
$L_ef_e=0$ gives $f_e(t)=\eta_ee^{-\kappa t}$ with $\eta_e=0$ for
$e\in\mathcal E_-$.
If $\Bin$ is forward, the vertex conditions are $\eta=\Mtr\eta$, so
$f=0$ when $I-\Mtr$ is invertible. Conversely, a null vector
$\eta\neq0$ of $I-\Mtr$ vanishes on source edges, because the
corresponding rows of $\Mtr$ do, and $f_e(t):=\eta_ee^{-\kappa t}$ is then
a nonzero element of $\widetilde H^1_{\Bin}(\Gamma)$ of zero energy.  This
proves the stated equivalence.  It applies to $K_1$ and $K_2$, and to
$C_V$ at vertices of in-degree one.  For $K_1$, and for $C_V$ at such
vertices, the incoming coefficients at each vertex sum to one. For each
$e\notin\mathcal E_-$, with $v=\operatorname{tail}(e)$,
\[
\sum_{\hat e}|\Mtr_{e,\hat e}|
=\sum_{\hat e}\Bin_v(e,\hat e)e^{-\kappa\ell_{\hat e}}
\le e^{-\kappa\ell_{\min}}<1.
\]
Every source row is zero, and hence
$\|\Mtr\|_\infty\le e^{-\kappa\ell_{\min}}<1$, whence
$\rho(\Mtr)<1$.  If instead $\Gamma$ is acyclic, ordering the edges
topologically makes $\Mtr$ strictly triangular, hence nilpotent, so
$\rho(\Mtr)=0$ for any $\Bin$, including $K_2$.

It remains to treat $C_V$ at a vertex with
$|\mathcal E^{\mathrm{in}}_v|\ge2$ on a possibly cyclic graph.  There $C_V$
is not forward, so $\Mtr$ is not available.  What continuity does give is that all traces at
$v$ share a common value, so for every $e$ with $v=\operatorname{tail}(e)$
and every $\hat e\in\mathcal E^{\mathrm{in}}_v$,
\begin{equation}
\label{eq:cv-onestep}
\eta_e=f_e(0)=f(v)=f_{\hat e}(\ell_{\hat e})=\eta_{\hat e}e^{-\kappa\ell_{\hat e}}.
\end{equation}
One such relation per edge suffices, by the following maximum principle.
Set $M_*:=\max_e|\eta_e|$ and suppose $M_*>0$.  Pick $e$ attaining the
maximum; its tail $v$ is not a source, since $\eta_e=0$ for
$e\in\mathcal E_-$, so $\mathcal E^{\mathrm{in}}_v\neq\varnothing$ and we
may choose $\hat e\in\mathcal E^{\mathrm{in}}_v$ in
\eqref{eq:cv-onestep}.  Then
$M_*=|\eta_e|=|\eta_{\hat e}|\,e^{-\kappa\ell_{\hat e}}
\le M_*e^{-\kappa\ell_{\min}}<M_*$, which is 
a contradiction; hence $M_*=0$, so $\eta=0$ and $f=0$.  
\end{proof}

\begin{proof}[Proof of \Cref{lem:q_RKHS_graph}]
We find constants $c_X,C_X>0$ with
${c_X\|f\|^2_{\widetilde H^1(\Gamma)}\le a_\Gamma^0(f,f)\le
C_X\|f\|^2_{\widetilde H^1(\Gamma)}}$ on $\widetilde H^1_X(\Gamma)$.
The upper bound follows from continuity of the endpoint traces and continuity of
${L_e:H^1(e)\to L_2(e)}$. If the lower bound failed, there would be
\(f_n\in\widetilde H_X^1(\Gamma)\) with
$\|f_n\|_{\widetilde H^1(\Gamma)}=1$ and 
${a_\Gamma^0(f_n,f_n)\longrightarrow0}$.
By compactness on the finitely many edges, after taking a subsequence,
\[
f_n\rightharpoonup f
\quad\text{in }\widetilde H^1(\Gamma),
\qquad
f_n\longrightarrow f
\quad\text{in }\bigoplus_{e\in\mathcal E}L_2(e).
\]
The trace constraints are closed, so \(f\in\widetilde H_X^1(\Gamma)\).
Moreover, $a_\Gamma^0(f_n,f_n)\to0$ and $\tau^2>0$ give
\(L_ef_{n,e}\to0\) in $L_2(e)$ for every $e$; the positivity of $c_v$
also gives $f_n(v)\to0$ at every source. Hence
\[
\partial_ef_{n,e}=L_ef_{n,e}-\kappa f_{n,e}
\longrightarrow-\kappa f_e
\quad\text{in }L_2(e),
\]
whereas weak $H^1$-convergence gives
$\partial_e f_{n,e}\rightharpoonup\partial_e f_e$. Uniqueness of the weak
limit therefore gives $\partial_e f_e=-\kappa f_e$. Thus
\(f_n\to f\) strongly in \(\widetilde H^1(\Gamma)\), so
\(\|f\|_{\widetilde H^1(\Gamma)}=1\). Continuity of the form gives
\(a_\Gamma^0(f,f)=0\), so positive definiteness yields \(f=0\), a
contradiction.

The space \(\widetilde H_X^1(\Gamma)
=\ker(\mv K^X\gamma)\) is closed in
\(\widetilde H^1(\Gamma)\), hence complete in the equivalent norm
\(\|f\|_q^2:=a_\Gamma^0(f,f)\). For
\((e,t)\in\{e\}\times[0,\ell_e]\), the one-dimensional Sobolev estimate and
the equivalence just proved give
$|f_e(t)|
\le C_e\|f_e\|_{H^1(e)}
\le C_e'\|f\|_q$.
Thus every edge-point evaluation is bounded, and the Riesz representation
theorem gives the reproducing kernel.
\end{proof}

\begin{proof}[Proof of \Cref{prop:transfer-construction}]
Any solution of \eqref{eq:system} satisfies \eqref{eq:edgewise-rep} on
each edge, and substituting the terminal traces into the second line of
\eqref{eq:system} gives $\eta_e=(\Mtr\eta)_e+\xi_e$ for
$e\notin\mathcal E_-$, while the third line gives $\eta_e=\xi_e$ for
$e\in\mathcal E_-$, where the corresponding row of $\Mtr$ vanishes.  Hence
$\eta$ solves \eqref{eq:fixedpoint}, which has the unique solution
$(I-\Mtr)^{-1}\xi$.  Conversely, defining $\eta$ by \eqref{eq:fixedpoint}
and $u$ by \eqref{eq:edgewise-rep} yields a field satisfying all three
lines of \eqref{eq:system}.  Since $\xi$ is a linear image of the jointly
Gaussian family $(\{u(v)\}_{v\in\mathcal V_-},\{\zeta_e\})$, both $\eta$
and $u$ are Gaussian; and $\xi_e$ depends on the noise only through
$\{\zeta_{\hat e}:\hat e\in\mathcal E^{\mathrm{in}}_{\operatorname{tail}(e)}\}$,
which with mutual independence of the source values gives the stated
independence.

For the Cameron--Martin space, $u$ is the image of 
${\bigl(\{u(v)\}_{v\in\mathcal V_-},
\{W_e\}_{e\in\mathcal E}\bigr)}$ under the
map defined by \eqref{eq:edgewise-rep} and \eqref{eq:fixedpoint}.  The
Cameron--Martin space of each $W_e$ is $H_0^1(e)$,
identified isometrically with $L_2(e)$ by differentiation. Using this
identification, the Cameron--Martin space of the family is
$\mathbb R^m\oplus\bigoplus_e L_2(e)$ with squared norm
${\sum_{v\in\mathcal V_-}\sigma_v^{-2}a_v^2+\sum_e\|h_e\|^2_{L_2(e)}}$, where $h_e\,dt$ replaces
$dW_e$.
Because the noise enters \eqref{eq:edge-sde} as $\tau^{-1}dW_e$, the
function generated on $e$ by the control $h_e$ satisfies
$L_ef_e=\tau^{-1}h_e$, and hence
$\sum_e\|h_e\|^2_{L_2(e)}=\tau^2\sum_e a_e^L(f_e,f_e)$,
which is the second term of \eqref{eq:q_semi}.
Consider the map $(a,h)\mapsto f$ sending
$((a_v)_{v\in\mathcal V_-},(h_e)_{e\in\mathcal E})$ to the
unique $f\in\widetilde H^1_{\Bin}(\Gamma)$ with
$f(v)=a_v$ for $v\in\mathcal V_-$ and
$L_ef_e=\tau^{-1}h_e$.  This is well defined and bijective: solving this equation
edgewise leaves the tail values $\eta_e$ free, and the vertex conditions
together with the prescribed source values determine $\eta$ uniquely
through \eqref{eq:fixedpoint}, which is solvable because $I-\Mtr$ is
invertible.  Under this bijection the squared norm above is exactly
$a_\Gamma^0(f,f)$. Moreover the map is the restriction to
Cameron--Martin spaces of the solution map itself: substituting
$a_v$ for the source value at each $v\in\mathcal V_-$ and
$h_e\,dt$ for $dW_e$ in
\eqref{eq:edgewise-rep} and \eqref{eq:fixedpoint} returns exactly $f$.
Since the Cameron--Martin space of a continuous linear image of a Gaussian
family is the image of that family's Cameron--Martin space, and the map here is a bijection, the
Cameron--Martin space of $u$ is $\widetilde H^1_{\Bin}(\Gamma)$ with the
inner product
$\langle f,g\rangle_{\mathcal H_u}=a_\Gamma^0(f,g)$.
\end{proof}

\begin{proof}[Proof of \Cref{prop:subdivision}]
The anchoring terms in \eqref{eq:q_semi} are unchanged, since the inserted
vertex is not a source. Write $e_1=[0,s]$ and
$e_2=[s,\ell_e]$ for the two new edges. Additivity gives
\[
a_e^L(f,g)
=a_{e_1}^L(f|_{e_1},g|_{e_1})+a_{e_2}^L(f|_{e_2},g|_{e_2}).
\]
At the inserted vertex, $C_V,K_1,K_2$ all equate the incoming
and outgoing labelled endpoint traces, exactly identifying the two pieces.
Under this identification the inner products agree, so
\Cref{lem:q_RKHS_graph} gives the same reproducing covariance kernel and
hence the same centred Gaussian law.
\end{proof}

\subsection{Proofs for Section~\ref{sec:acyclic}}
\begin{lemma}
\label{lem:gaussian-measurable-linear}
Let \(X\in\mathbb R^m\) and \(Y\in\mathbb R^n\) be jointly Gaussian and
centred. If
\(Y\) is \(\sigma(X)\)-measurable, then a deterministic 
\(B\in\mathbb R^{n\times m}\) exists such that \(Y=BX\) a.s.
\end{lemma}

\begin{proof}
Since \((X,Y)\) is jointly Gaussian, the conditional law of \(Y\) given \(X\) is Gaussian
with mean affine in \(X\) and covariance independent of \(X\). If \(Y\) is
\(\sigma(X)\)-measurable, then this conditional law is almost surely a Dirac mass, so its
conditional covariance vanishes. Hence \(Y=\mathbb E[Y\mid X]=BX\) a.s.~for some
deterministic matrix \(B\).
\end{proof}

\begin{proof}[Proof of \Cref{lem:graph-OU-form}]
Fix \(v\notin\mathcal V_-\). The incoming trace vector and the outgoing
initial-value vector are jointly centred Gaussian. By
\Cref{def:proper-global-OU}, the latter is measurable with respect to the
former. Hence \Cref{lem:gaussian-measurable-linear} gives a deterministic
matrix \(B_v\) such that
\[
\bigl(\eta_e\bigr)_{e\in\mathcal E_v^{\mathrm{out}}}
=
B_v\bigl(u_{\hat e}(\ell_{\hat e})\bigr)_{\hat e\in\mathcal E_v^{\mathrm{in}}}
\qquad\text{a.s.}
\]
These are exactly the forward vertex relations in \eqref{eq:system}, with
\(\Bin_v(e,\hat e)=(B_v)_{e,\hat e}\).

Using the source-edge notation of \Cref{sec:graphs}, acyclicity
makes
\(\{\eta_{e_v}\}_{v\in\mathcal V_-}\) and
\(\{W_e\}_{e\in\mathcal E}\) a complete set of independent drivers. Replace
\(dW_e\) by \(h_e(t)\,dt\), with \(h_e\in L_2(e)\), and replace
\(\eta_{e_v}\) by \(a_v\in\mathbb R\). The resulting deterministic path
satisfies
\[
f_e(t)=e^{-\kappa t}b_e+\tau^{-1}\!\int_0^t e^{-\kappa(t-s)}h_e(s)\,ds,
\qquad t\in[0,\ell_e],
\]
where \(b_{e_v}=a_v\) at a source and, at every interior vertex,
$\bigl(b_e\bigr)_{e\in\mathcal E_v^{\mathrm{out}}}
=
B_v\bigl(f_{\hat e}(\ell_{\hat e})\bigr)_{\hat e\in\mathcal E_v^{\mathrm{in}}}$.
Thus every controlled path lies in \(\widetilde H^1_{\Bin}(\Gamma)\).
Conversely, any \(f\) in this space is obtained by taking
\(a_v=f(v)\) and \(h_e=\tau L_ef_e\).

It remains to compute the inner product. For arbitrary
$f,g\in\widetilde H^1_{\Bin}(\Gamma)$, the polarized contribution of the
control on edge $e$ is
$
\tau^2a_e^L(f_e,g_e).
$
The source variable
\(\eta_{e_v}\sim\pN(0,\sigma_v^2)\) contributes
\(\sigma_v^{-2}f(v)g(v)\). Interior initial values add no term because the
vertex relations determine them from upstream traces. Summing over the
independent drivers gives
\[
\langle f,g\rangle_{\Gamma}^{\mathrm{OU}}
=
\tau^2\sum_{e\in\mathcal E}a_e^L(f_e,g_e)
+
\sum_{v\in\mathcal V_-}\sigma_v^{-2}\,f(v)g(v),
\]
which proves the claim.
\end{proof}

\begin{proof}[Proof of \Cref{cor:forward-substitution}]
Order the edges topologically. By \eqref{eq:transfer-matrix},
\(\Mtr_{e,\hat e}\) can be nonzero only if \(\hat e\) precedes \(e\).
Thus \(\Mtr\) is strictly triangular, so
\[
    \Mtr^{|\mathcal E|}=0,
    \qquad
    (I-\Mtr)^{-1}=\sum_{k=0}^{|\mathcal E|-1}\Mtr^k.
\]
In particular, \(I-\Mtr\) is invertible and
\Cref{prop:transfer-construction} applies. Reading \eqref{eq:fixedpoint} in
the same ordering determines each outgoing initial value from the terminal
traces immediately upstream. This is the measurability condition in
\Cref{def:proper-global-OU}, with
\(B_v=(\Bin_v(e,\hat e))\); its remaining conditions hold by construction.
\end{proof}

We next record a transfer representation used in the covariance proofs.
\begin{lemma}\label{lem:OU-transfer-representation}
Let \(\Gamma\) be a finite acyclic directed metric graph and let \(u\) be a
proper global OU process on \(\Gamma\).  Then, for every
\(y\in\Gamma^\circ\),
\[
u(y)=\sum_{j\in\mathcal J(y)} \mathsf A(s_{-,j},y)\,Z_j
+\tau^{-1}\!\int_{\Lambda^\uparrow(y)} \mathsf A(z,y)\,dW(z),
\]
with \(\mathsf A\) as in \eqref{eq:transferFactor}.
\end{lemma}

\begin{proof}
Since \(\Gamma\) is finite and acyclic, the edges of
\(\Lambda^\uparrow(y)\) admit a topological ordering from the inflow
boundary toward \(y\).  On each edge the representation of
\Cref{def:proper-global-OU} gives
\[
u_e(t)=e^{-\kappa (t-s)}u_e(s)+\tau^{-1}\int_s^t e^{-\kappa (t-\xi)}\,dW_e(\xi),
\qquad 0\le s\le t\le \ell_e,
\]
and at each interior vertex \Cref{lem:graph-OU-form} gives
\(u(\underline v)=\sum_{\bar v}\Bin_v(\underline v,\bar v)u(\bar v)\) a.s.
Both relations are linear, so \(u(y)\) is a linear functional of the source
values and the edge noises.  Iterating them along a directed path from a
driver to \(y\) multiplies the exponential decay accumulated along the path
by the routing coefficient at each vertex the path traverses; a driver that
reaches \(y\) along several distinct paths contributes once for each, and
the coefficients add.  The total coefficient of \(Z_j\) is therefore the sum
over paths \eqref{eq:transferFactor}, that is \(\mathsf A(s_{-,j},y)\), and
that of \(dW(z)\) is \(\tau^{-1}\mathsf A(z,y)\), the factor \(\tau^{-1}\)
coming from the edgewise representation.  Summing gives the claim.  On a
tree at most one path joins any pair of points and each sum has a single
term.
\end{proof}

\begin{proof}[Proof of \Cref{thm:cov-kernel}]
By \Cref{lem:OU-transfer-representation}, \(u(x)\) and \(u(y)\) are linear
functionals of the source values \(Z_j\) and the edge noises, with
coefficients \(\mathsf A(s_{-,j},\cdot)\) and
\(\tau^{-1}\mathsf A(z,\cdot)\).  By \eqref{eq:system} the \(Z_j\) are
mutually independent and independent of the noises, and the noises are
independent across disjoint portions of \(\Gamma\).  Taking covariances
therefore pairs each driver with itself: the source terms give
\(\sum_j\sigma^2_{s_{-,j}}\mathsf A(s_{-,j},x)\mathsf A(s_{-,j},y)\), and
the It\^o isometry gives
\(\tau^{-2}\int\mathsf A(z,x)\mathsf A(z,y)\,dz\).  A driver contributes
only if it reaches both points, so the sum is over
\(\mathcal J(x)\cap\mathcal J(y)\) and the integral over
\(\mathcal C^\uparrow(x,y)\); elsewhere one of the two factors vanishes.
\end{proof}
\begin{proof}[Proof of \Cref{cor:variance-formula}]
The first identity is \eqref{eq:cov-kernel} with \(x=y=s\).  For the
second, let \(x\leadsto y\) lie on a common edge. Then no vertex is traversed
between them, so \(\mathsf A(x,y)=e^{-\kappa d(x,y)}\) and
${u(y)=e^{-\kappa d(x,y)}u(x)+\tau^{-1}\int_x^y e^{-\kappa d(z,y)}dW(z)}$,
the integral being over the segment from \(x\) to \(y\).  That segment is
not upstream of \(x\), since \(\Gamma\) is acyclic, so the two terms are
independent and
\(\Var(u(y))=e^{-2\kappa d(x,y)}\Var(u(x))
+\tau^{-2}\int_0^{d(x,y)}e^{-2\kappa r}dr\), which is
\eqref{eq:within-edge-variance}.
\end{proof}
\begin{proof}[Proof of \Cref{thm:cov-lca}]
We first check that $\mathcal C^\uparrow(s,t)$, when non-empty, has a
unique maximal element.  Write $P$ for the undirected path joining $s$ and
$t$, and let $z\in\mathcal C^\uparrow(s,t)$.  Since $\Gamma$ is a tree, the
directed paths from $z$ to $s$ and from $z$ to $t$ are the unique
undirected ones. Let $a_z$ be their last common point. Then
$z\leadsto a_z$, $a_z\leadsto s$, and $a_z\leadsto t$, and the remaining
path segments from $a_z$ have disjoint interiors. Moreover $a_z\in P$,
because the paths from $a_z$ to $s$ and from $a_z$ to $t$ are edge-disjoint
apart from $a_z$, so their concatenation is the
undirected $s$--$t$ path.  At most one point of $\mathcal C^\uparrow(s,t)$
can lie on $P$: if $m\neq m'$ both did, with $m$ between $s$ and $m'$ along
$P$, then $m'\leadsto s$ would force $m'\leadsto m$ and $m\leadsto t$ would
force $m\leadsto m'$, a directed cycle, which a tree does not admit.
Hence $a_z$ is one and the same point $a$ for every
$z\in\mathcal C^\uparrow(s,t)$, and $a$ is the maximum of
$\mathcal C^\uparrow(s,t)$: it lies in the set, and $z\leadsto a$ for every
$z$ in it.

Consequently $\mathcal C^\uparrow(s,t)=\Lambda^\uparrow(a)$ and
\(\mathcal J(s)\cap\mathcal J(t)=\mathcal J(a)\).  For \(z\leadsto a\) the
directed path from \(z\) to \(s\) is the concatenation of those from \(z\)
to \(a\) and from \(a\) to \(s\), hence passes through \(a\), so
\(\mathsf A(z,s)=\mathsf A(z,a)\mathsf A(a,s)\), and likewise for \(t\).
Substituting in \eqref{eq:cov-kernel} lets \(\mathsf A(a,s)\mathsf A(a,t)\)
be taken out of both the sum and the integral, leaving
\begin{align*}
r(s,t)&=\mathsf A(a,s)\mathsf A(a,t)\Bigl\{
\sum_{j\in\mathcal J(a)}\sigma^2_{s_{-,j}}\mathsf A(s_{-,j},a)^2
+\tau^{-2}\!\int_{\Lambda^\uparrow(a)}\!\mathsf A(z,a)^2dz\Bigr\}\\
&=r(a,a)\,\mathsf A(a,s)\,\mathsf A(a,t),
\end{align*}
the brace being \(r(a,a)\) by \Cref{cor:variance-formula}.  If
\(\mathcal C^\uparrow(s,t)=\varnothing\) then no driver reaches both points
and \eqref{eq:cov-kernel} is empty, so \(r(s,t)=0\).  If \(t\leadsto s\)
then \(t\) is itself the maximal common ancestor and \(\mathsf A(t,t)=1\),
giving \(r(s,t)=r(t,t)\mathsf A(t,s)\).
\end{proof}
\begin{proof}[Proof of \Cref{cor:stationary-variance-K2}]
Set \(r_*:=\sigma^2=(2\kappa\tau^2)^{-1}\), the stationary variance. If
\(x\leadsto y\) lie on the same edge, \eqref{eq:within-edge-variance} gives
\[
r(y,y)=e^{-2\kappa d(x,y)}\,r(x,x)+\frac{1-e^{-2\kappa d(x,y)}}{2\kappa\tau^2},
\]
so \(r(x,x)=r_*\) implies \(r(y,y)=r_*\). Thus stationarity of the variance is preserved along each edge.
Now let \(v\) be an interior vertex, and suppose that \(r(\bar v,\bar v)=r_*\) for all \(\bar v\in\mathcal V_v^{\mathrm{in}}\). For each \(\underline v\in\mathcal V_v^{\mathrm{out}}\), the \(K_2\)-condition gives
$u(\underline v)=\sum_{\bar v\in\mathcal V_v^{\mathrm{in}}}\beta_v(\underline v,\bar v)\,u(\bar v)$ and 
$\sum_{\bar v\in\mathcal V_v^{\mathrm{in}}}\beta_v(\underline v,\bar v)^2=1$.
Because $\Gamma$ is a tree, distinct incoming traces into $v$ have pairwise
disjoint upstream driver sets. Hence the variables
${\{u(\bar v):\bar v\in\mathcal V_v^{\mathrm{in}}\}}$ are independent:
\Cref{lem:OU-transfer-representation} expresses the traces through disjoint
families of source values and edge noises, which are mutually independent by
\eqref{eq:system}. Hence
${r(\underline v,\underline v)
=
\sum_{\bar v\in\mathcal V_v^{\mathrm{in}}}\beta_v(\underline v,\bar v)^2\,r(\bar v,\bar v)
=
r_*}$.
Thus stationarity of the variance is also preserved across every interior vertex.

Finally, choose a topological ordering of the edges from inflow leaves to the terminal vertices. The inflow values have variance \(r_*\) and repeated application of the two propagation steps above therefore yields
\(r(t,t)=r_*\) for every \(t\in\widetilde\Gamma\).
\end{proof}

\subsection{Proofs for Sections~\ref{sec:relation} and~\ref{sec:inference}}
\begin{proof}[Proof of \Cref{prop:sym-vs-dir}]
By \Cref{lem:edge-ips}, for each edge
\[
\tau^2a_e^L(f_e,g_e)
=\tau^2a^A_e(f_e,g_e)
+\kappa\tau^2\bigl\{f_e(\ell_e)g_e(\ell_e)-f_e(0)g_e(0)\bigr\}.
\]
Summing over $e\in\mathcal E$ and adding the anchoring term gives
\eqref{eq:q_semi} on the left and
$\langle f,g\rangle^{\mathrm{sym}}_\Gamma$ plus the endpoint sum on the
right.  Since $f,g\in\widetilde H^1_{C_V}(\Gamma)$ are single-valued at
every vertex, each edge end at $v$ contributes $\pm\kappa\tau^2f(v)g(v)$,
with a plus sign when $v$ is the head of that edge and a minus sign when it
is the tail.  Collecting the $2|\mathcal E|$ endpoint terms by vertex, the
coefficient of $\kappa\tau^2 f(v)g(v)$ is the number of edges entering $v$
minus the number leaving it, which is
$|\mathcal E^{\mathrm{in}}_v|-|\mathcal E^{\mathrm{out}}_v|$.  This is
\eqref{eq:sym-vs-dir}.
\end{proof}

\begin{proof}[Proof of \Cref{prop:tailup}]
Constant variance is \Cref{cor:stationary-variance-K2}.  If $t\leadsto s$
then $t$ is the last common ancestor of $s$ and $t$, so \Cref{thm:cov-lca}
gives $r(s,t)=r(t,t)\mathsf A(t,s)=\sigma^2\mathsf A(t,s)$, and
$\mathsf A(t,s)$ is the single product \eqref{eq:transferFactor} along the
unique directed path, whose routing coefficients are
$\sqrt{p_{v,\bar v(t,s)}}$ under $K_2$.
The case $s\leadsto t$ follows by exchanging $s$ and $t$ and using
covariance symmetry.
Now suppose neither point is upstream of the other, and let
$z\in\mathcal C^\uparrow(s,t)$.  Because $\Gamma$ is oriented along the
flow, every vertex other than the outlet has exactly one outgoing edge, so
the directed path leaving $z$ never branches and runs to the outlet.  Both
$s$ and $t$ are downstream of $z$ and therefore lie on it,
hence one precedes the other and one of the two points is upstream of the
other, contrary to assumption.  So $\mathcal C^\uparrow(s,t)=\varnothing$
and $r(s,t)=0$ by \Cref{thm:cov-lca}.
\end{proof}

\begin{proof}[Proof of \Cref{prop:taildown}]
Reversal exchanges $\mathcal E^{\mathrm{in}}_v$ and
$\mathcal E^{\mathrm{out}}_v$ at every $v$, so each interior vertex of
$\Gamma^R$ has exactly one incoming edge. Then
$p_{v,\hat e}=1$ and the three coefficients in
\eqref{eq:vertex-condition} all equal $1$.  The outlet of $\Gamma$ becomes
the unique source of $\Gamma^R$, and since it is a leaf it is an inward
leaf of $\Gamma^R$, so
\Cref{ass:standing} holds and \Cref{cor:stationary-variance-K2}
gives stationary variance at every labelled edge point. Since
all routing coefficients equal $1$, all incident traces agree, so the field
descends to $\Gamma$. With all routing
coefficients equal to $1$,
\eqref{eq:transferFactor} reduces to $\mathsf A(x,y)=e^{-\kappa d(x,y)}$
whenever $x$ precedes $y$ in $\Gamma^R$.
Now, let $s,t\in\Gamma$ and let $a$ be their last common ancestor in
$\Gamma^R$, which exists because $\Gamma^R$ is a tree with a single source
reaching every point.  \Cref{thm:cov-lca} gives
$r(s,t)=\sigma^2e^{-\kappa(d(a,s)+d(a,t))}$.  If one of $s,t$ precedes the
other in $\Gamma^R$ then $a$ is that point and the exponent is $d(s,t)$.
Otherwise $a$ is the vertex at which the two directed paths from the
outlet separate and the
path from $s$ to $t$ in $\Gamma$ passes through it, so again
$d(a,s)+d(a,t)=d(s,t)$.
\end{proof}


\end{document}